\documentclass[11pt,a4paper]{article}

\usepackage[utf8]{inputenc}        
\usepackage[T1]{fontenc}           
\usepackage{mathpazo}              
\usepackage{microtype}             
\usepackage{setspace}              
\usepackage{amsmath, amssymb, amsthm} 
\usepackage{graphicx}              
\usepackage{booktabs}              
\usepackage{enumitem}              
\usepackage{xcolor}                
\usepackage[hidelinks]{hyperref}   

\usepackage[
    left=1.2in,
    right=1.2in,
    top=1.0in,
    bottom=1.2in
]{geometry}

\title{\textbf{Stochastic gradient descent on the epigenetic landscape: a unified framework for cellular plasticity, tumor heterogeneity, and the asymptotic irrelevance of fitness}}
\author{
    Artur César Fassoni\textsuperscript{1,2,*} \\
    \small \textsuperscript{1}Instituto de Matemática e Computação, Universidade Federal de Itajubá, Itajubá, Brazil \\
    \small \textsuperscript{2}Carl Gustav Carus School of Medicine, Technische Universität Dresden, Dresden, Germany \\
    \small \textsuperscript{*}\texttt{fassoni@unifei.edu.br}
}

\date{
}

\newtheorem{theorem}{Theorem}

\newtheorem{lemma}{Lemma}

\theoremstyle{remark}
\newtheorem{remark}{Remark}
\theoremstyle{definition}

\newcommand{\one}{\mathbf 1}
\newcommand{\abs}[1]{\lvert #1\rvert}
\newcommand{\R}{\mathbb{R}}

\begin{document}

\maketitle

\begin{abstract}

Phenotypic plasticity, the ability of cells to switch between states, is central to development, differentiation, and therapy resistance. Although it is modeled at several scales, from compartmental ODEs to phenotype-structured PDEs and single-cell stochastic equations, a framework connecting these descriptions is missing. We present such a framework. Starting from a general $n$-compartment ODE model encompassing nonlinear growth and linear transitions between phenotypes, we show, with a new, elementary and generalizable proof of a recent theorem, that under uniform competition, the long-term population distribution is solely governed by transition rates. All phenotypes become selectively neutral at saturation, and the imprint of fitness differences during growth fades at an explicit rate. Restricting transitions to neighboring states transforms the model into a discretization of a phenotype-structured reaction-diffusion-advection PDE. In this continuum model, diffusion and advection are identified from switching rates, and fitness remains asymptotically irrelevant. Interpreting the advection velocity as the negative gradient of an effective epigenetic potential transforms the PDE into a Fokker--Planck equation and converts single-cell trajectories into stochastic gradient descent (SGD) in the Langevin sense on the phenotypic landscape. Non-local transitions, such as mutations, are incorporated via an integro-differential term, yielding a unified model with reaction, gradient flow, diffusion, and jumps. State-dependent noise reshapes the effective landscape without altering the underlying potential. This gives cancer a route to elevated plasticity that static, single-cell snapshots cannot distinguish from a changed landscape. This framework provides a physical interpretation of Waddington's landscape, where cells perform SGD, and cancer corresponds to a corrupted landscape.
\end{abstract}


\section{Introduction}

Cellular plasticity, broadly defined, is the ability of a cell to change its state or phenotype without altering its genome \cite{grafenver2009,merrellstanger2016}. By \emph{state} we mean a typically heritable, self-sustaining pattern of gene expression, such as the ones that distinguish a stem cell from a differentiated neuron or a drug-sensitive from a drug-resistant cell; it persists across cell divisions and on the time scale of interest without external maintenance, but is not fixed by the DNA sequence itself. This property is fundamental to multicellular organisms: almost all cells of an individual share essentially the same DNA, yet display a wide range of phenotypes. From fertilization onward, differentiation is regulated by epigenetic changes that channel initially totipotent cells toward specialized fates \cite{bird2007}, and these epigenetic barriers are what keep tissue organization stable and prevent uncontrolled de-differentiation. Waddington captured this with his ``epigenetic landscape" metaphor \cite{waddington1957}: a state corresponds to a valley in this landscape, and plasticity to the topography of slopes, barriers, and routes that favors certain transitions between valleys and blocks others. Widening a valley, lowering a barrier, or opening a new route lets cells reach states that would normally be closed to them.

At the molecular level, epigenetic alterations are largely mediated by histone modifications \cite{jenuwein2001,kouzarides2007,allisjenuwein2016}, which determine which genes remain accessible for transcription. Stable phenotypes are self-sustaining configurations of this regulatory machinery, into which a cell settles and from which it escapes only through fluctuations. These configurations are not thermodynamic equilibria, since histone marks are written and erased by enzymes that consume energy \cite{dodd2007,michieletto2016,alarcon2026}, but their dynamics can be summarized by an effective potential, or quasi-potential \cite{wang2015,zhou2016}. Waddington's landscape then has a physical counterpart: its valleys are minima of this potential, and transitions between phenotypes are stochastic excursions between them.

Cancer disrupts this order \cite{cagan2026}. Tumors accumulate genetic and phenotypic alterations \cite{lv2026}, and intratumor heterogeneity has traditionally been understood as the product of somatic evolution \cite{marusykpolyak2010,greavesmaley2012}: reiterated rounds of clonal expansion, genetic diversification and selection shape the clonal architecture of a tumor \cite{vendramin2021,gerlinger2012}. Increased cellular plasticity, recently recognized as one of the hallmarks of cancer \cite{hanahan2022}, adds a second, non-genetic source of diversity \cite{meachammorrison2013}: malignant cells gain access to alternative phenotypes that are not encoded by new mutations \cite{yuanstanger2019,flavahan2017}.

The clinical relevance of this second source is most evident in therapy resistance. In several tumor types, including breast, prostate, lung, melanoma, and glioblastoma, subpopulations of cells resist chemotherapeutics or targeted therapies without any underlying genetic change \cite{sharma2010,marinedawson2020,pisco2013,pisco2015}. In immunotherapies such as CAR-T cells, a major resistance mechanism is antigen loss, in which tumor cells stop expressing the target antigen through epigenetic or post-transcriptional mechanisms and settle into stable phenotypic states invisible to the therapy \cite{shahfry2019,santurio2024,majznermackall2019}. Resistance can therefore arise by two routes: the selection of pre-existing resistant clones, or the emergence of resistance de novo from drug-tolerant persister cells \cite{dagogojackshaw2018}. Heterogeneity itself is, in turn, associated with prognosis \cite{ferrallfairbanks2019,morris2016}.

Because genetic and non-genetic changes jointly shape tumor evolution and its response to treatment, it is desirable to describe them within a single framework. The landscape picture introduced above offers one: if we relax the definition of plasticity slightly and conceptualize it as the ability to transition among states in an extended space that includes both phenotypes and genotypes, mutations correspond to discrete jumps in that space. Genetic instability seeds cells carrying driver mutations \cite{nowell1976,hanahanweinberg2011}, and the elevated epigenetic plasticity of tumors lets these founder cells reach distant phenotypic states that are normally off-limits to healthy cells \cite{hanahan2022}.

In this picture, the difference between healthy and malignant tissue is one of landscape geometry and noise: in normal tissue, cells rest in deep attractors under low noise, whereas in tumors transcriptional programs are noisier and cells move more readily between attractors \cite{huang2009,marusyk2012,feinbergirizarry2010}. The framework developed here turns this qualitative picture into a quantitative one: selection enters through the vital dynamics (birth and death), plasticity through the transport terms, noise as a diffusivity, and mutation through non-local jumps. It allows us to ask when each of these ingredients shapes the phenotypic composition of a tumor, and on which time scale.

Mathematical models of these phenomena roughly follow the arc of increasing complexity we develop below. Compartmental, ordinary differential equation (ODE) models describe each cellular state as a compartment with its own vital dynamics, and plasticity as transition rates between compartments; they have been used for differentiation hierarchies, therapy resistance and immune escape \cite{marciniakczochra2009,michor2005,roeder2006,fassoni2018,strobl2021,zhoutraulsen2019,gavrilova2026,santurio2024,zhou2013,gunnarsson2020}. Driven by single-cell data and dimensionality reduction \cite{coifmanlafon2006,haghverdi2016}, another perspective holds that phenotypic states form a continuum \cite{trapnell2014}: between a stem cell and a progenitor, or between drug-sensitive and drug-resistant states, lie intermediate phenotypes. The continuum limit of compartmental models has accordingly grown into a substantial body of work on phenotype-structured partial differential equations (PS-PDEs) \cite{perthame2007,lorz2011,villa2025,clairambault2022,agostinelli2026,chisholm2015,lorenzi2015,lorenzi2016,clairambaultpouchol2019,cho2018,singh2025population,de2024analysis}. At the level of individual cells, stochastic differential equations (SDEs), notably Ornstein--Uhlenbeck (OU) processes, describe phenotypic change as a restoring force perturbed by noise \cite{desouzasilva2023,kessler2022,parklevine2025}.

Alongside this arc, a parallel tradition in statistical physics formalizes Waddington's landscape as a potential, or quasi-potential, of the underlying stochastic dynamics \cite{huang2005,ao2007,wang2011,wang2015,zhou2016,coomer2022}; dynamical-systems and geometric approaches relate fate decisions to bifurcations and low-dimensional landscapes \cite{ferrell2012,moris2016,corson2012,rand2021,saez2022}; and more recent work derives effective landscapes from chromatin dynamics, Bayesian decision-making or evolvability \cite{alarcon2026,entropy2026bayesian,jimenezsanchez2026}. These works describe plasticity in different ways: discrete or continuous, phenomenological or mechanistic, at the level of the population or of single stochastic cells. Yet they are seldom connected explicitly. Population models usually posit transition rates or transport terms phenomenologically, whereas landscape models derive them from an underlying mechanism but tend to leave out selection and mutation.

Here we combine results from applied mathematics, statistical physics and quantitative biology into a framework that describes the epigenetic and genetic evolution of cell populations in a continuous state space, with an emphasis on the mechanistic connections between these approaches. The article is both a \emph{review} and a \emph{synthesis}: it follows a single line from compartmental ODE models, through PS-PDEs, to SDEs for the trajectories of individual cells. Section~\ref{sec:model} develops this construction step by step; each of its subsections ends with a \emph{Relation to prior work} paragraph that credits the work it builds on and says where our construction differs. The literature is large, and we focus on the work most directly relevant to the mathematical structure developed here; readers looking for broader perspectives can consult the tutorial of Lorenzi et al.~\cite{villa2025} on phenotype-structured PDEs, and the works of Alvarez et al.~\cite{clairambault2022} on plasticity and evolution and of de Souza Silva et al.~\cite{desouzasilva2023} on the Ornstein--Uhlenbeck framework in quantitative genetics.

Following this line, our main contributions are:

\begin{itemize}
\item A new, generalizable proof, under different hypotheses and with explicit estimates, of a recent theorem of Giaimo et al.~\cite{giaimo2025} stating that, in \emph{uniformly competitive} compartmental models, fitness is asymptotically irrelevant, i.e., the long-time phenotypic composition does not depend on the fitness differences among states. Our version gives an explicit estimate of how fast fitness differences are forgotten and a bound on the total selection that the phenotypic composition can undergo. We also state the conditions required (connected plasticity, persistent proliferation) and give a counterexample showing that uniform competition cannot be relaxed to weighted competition (Sections~\ref{sec:compartiment} and~\ref{sec:uniform_competition}, Appendix~A).

\item An interpretation of compartmental models with local transitions as discretizations of a phenotype-structured reaction--diffusion--advection PDE, with diffusion and advection identified from the switching rates (Section~\ref{sec:continuous}), and the extension of the asymptotic irrelevance of fitness under uniform competition to this continuum analogue (Section~\ref{sec:pdeasymp}).

\item An interpretation of this PDE as the Fokker--Planck equation of single cells performing a stochastic gradient descent on an effective epigenetic potential, with noise as an effective temperature (Sections~\ref{sec:stochastic} and~\ref{sec:gradient}), and a unified model of reaction, gradient flow, diffusion, and jumps, whose structure is motivated by the Kramers--Moyal expansion and Pawula's theorem (Section~\ref{sec:nonlocal}).

\item A derivation of classical phenotypic distributions (uniform, exponential, Gaussian, heavy-tailed) from the potential and the diffusivity, together with the observation that state-dependent noise reshapes the effective landscape independently of the bare potential (Section~\ref{sec:examples}).
\end{itemize}

Section~\ref{sec:conclusion} discusses the scope of these results and some open problems.

\section{From discrete switching between cell states to stochastic gradient descent in the epigenetic landscape}
\label{sec:model}

\begin{figure}[htbp]
\centering
\includegraphics[width=0.94\textwidth]{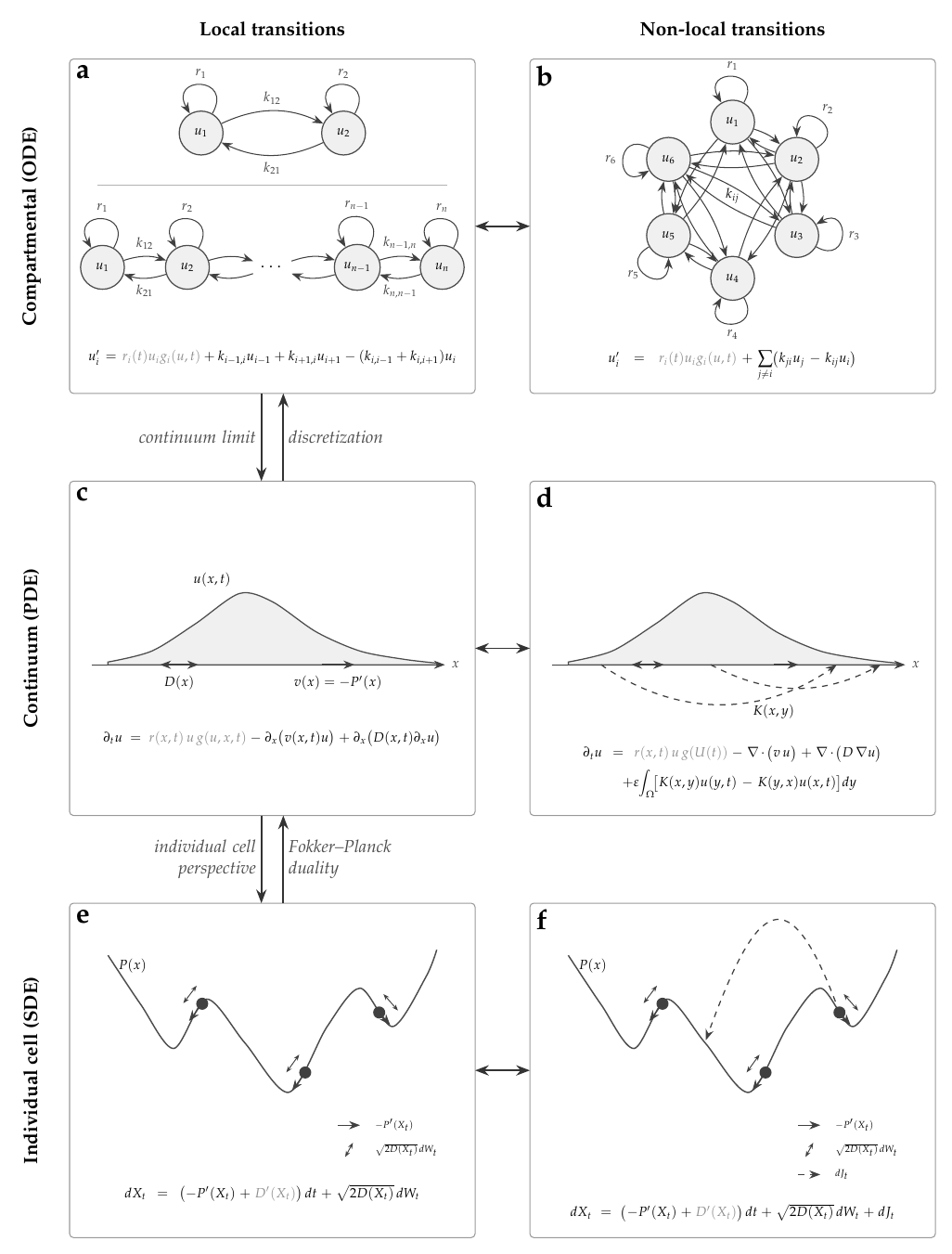}
\caption{\footnotesize\linespread{1}\selectfont\textbf{Roadmap of the modeling framework.} Rows correspond to the level of description (compartmental ODEs, top; continuum PDEs, middle; individual-cell SDEs, bottom); columns to whether transitions are local (left) or non-local (right). The reaction (vital-dynamics) term is shown in gray to indicate that it becomes asymptotically negligible under uniform competition (Theorems~\ref{thm:disc} and \ref{thm:cont}) or vanishes at equilibrium. Horizontal double-headed arrows connect the local and non-local versions of the same level of description; vertical arrows connect the continuum limit/discretization (rows 1--2) and the Fokker--Planck/individual-trajectory duality (rows 2--3). \textbf{(a)} Two-compartment model, eq.~\eqref{eq:modelo2d} (top), and its local $n$-compartment generalization, eq.~\eqref{eq:modelondloc} (bottom). \textbf{(b)} The same construction without the nearest-neighbor restriction, eq.~\eqref{eq:modelond}. \textbf{(c)} Continuum limit of (a): the advection--diffusion equation \eqref{eq:eqlinssprob}, with advection velocity interpreted as a negative gradient, $v(x) = -P'(x)$. \textbf{(d)} Adding a non-local jump term with kernel $K(x,y)$ yields the unified model \eqref{eq:unified}. \textbf{(e)} The single-cell SDE \eqref{eq:SDEgradient}: a stochastic gradient descent on the potential $P(x)$, with drift toward the nearest minimum and noise that lets a cell escape shallow ones by crossing the barriers between them. \textbf{(f)} Adding jumps $dJ^{\varepsilon}_t$ (eq.~\eqref{eq:unifiedSDE}) lets a cell relocate directly across the landscape.}
\label{fig:roadmap}
\end{figure}

\subsection{Compartmental ODE models for cellular states}
\label{sec:compartiment}

We start with the compartmental perspective, in which each cellular state, a differentiation stage, an epigenetic configuration, or a drug-sensitive versus drug-resistant phenotype, is a discrete compartment whose abundance evolves according to an ODE for proliferation, death, and transitions between states. This is the simplest mathematical representation of phenotypic plasticity, and the basis for the more general and mechanistic continuous descriptions that follow (Figure~\ref{fig:roadmap}a,b).

\medskip
\noindent\textbf{A simple two-compartment framework.} To fix ideas, consider the following simple model of two interconverting cellular states:
\begin{equation}\label{eq:modelo2d}
\begin{aligned}
    u_1' &= -k_{12} u_1 + k_{21} u_2 + r_1 u_1 \left(1 - \frac{u_1+u_2}{K}\right), \\
    u_2' &= -k_{21} u_2 + k_{12} u_1 + r_2 u_2 \left(1 - \frac{u_1+u_2}{K}\right) - d_2(t) u_2.
\end{aligned}
\end{equation}
Here $u_1(t)$ and $u_2(t)$ denote the numbers of cells in states $1$ and $2$, respectively; $r_1, r_2 > 0$ are intrinsic, net proliferation rates; $K$ is a carrying capacity; and $d_2(t) \ge 0$ is a time-dependent death rate acting on state $2$. Constants $k_{12}$ and $k_{21}$ are transition rates between states $1$ and $2$. These switching terms represent phenotypic plasticity.

Model \eqref{eq:modelo2d}, simple as it is, has been used in several biological contexts \cite{gavrilova2026,strobl2021,fassoni2018,santurio2024}. For instance, in drug resistance, $u_1$ denotes drug-resistant cells and $u_2$ drug-sensitive cells. The additional mortality $d_2(t)$ represents the effect of therapy. Cells can acquire or lose the resistant phenotype via $k_{12}$ and $k_{21}$, and therapy selects the resistant population \cite{strobl2021}. In models for chronic myeloid leukemia treatment with TKIs, $u_1$ and $u_2$ represent quiescent and proliferative leukemic stem cells, with $r_1 = 0$ reflecting quiescence and $d_2(t)$ capturing TKI-induced death of cycling cells \cite{fassoni2018}. In CAR-T therapy, $u_1$ represents antigen-negative tumor cells and $u_2$ antigen-positive cells, while the transitions $k_{ij}$ reflect epigenetic modulation of antigen expression \cite{santurio2024}.

In all these examples the model has the same two components: \emph{vital dynamics} (proliferation and competition) and \emph{transition dynamics} (phenotypic switching).

\medskip
\noindent\textbf{A general model with $n$ compartments.} We now generalize model \eqref{eq:modelo2d} to an arbitrary number of states and, later, to a continuous phenotype space. Consider $n$ compartments, where $u_i(t)$ is the abundance in state $i$, governed by
\begin{equation}\label{eq:modelond}
    u_i' = r_i(t)\, u_i \, g_i(u,t) + \sum_{j \neq i} \bigl( k_{ji} u_j - k_{ij} u_i \bigr), \qquad i = 1, \dots, n.
\end{equation}
Here, $r_i(t) \ge 0$ is the net proliferation rate, possibly time-dependent but uniformly bounded ($r_i(t) \le r_{\max}$ for some constant $r_{\max}>0$); $g_i(u,t)$ limits growth (density dependence, competition); its specific forms include logistic, Gompertz, or competitive Lotka--Volterra forms, and will be discussed in Section~\ref{sec:uniform_competition}; finally, $k_{ij} \ge 0$ is the transition rate from compartment $i$ to compartment $j$, encoding phenotypic plasticity.

\medskip
\noindent\textbf{Matrix form.} It is convenient to rewrite the $n$-compartment model \eqref{eq:modelond} in matrix notation. Let
\[
u(t) = \begin{bmatrix} u_1(t) \\ \vdots \\ u_n(t) \end{bmatrix}
\]
be the column vector of cell numbers in each state. Define the diagonal matrices
\[
R(t) = \operatorname{diag}\bigl(r_1(t), \dots, r_n(t)\bigr), \qquad
\operatorname{diag}(u) = \operatorname{diag}\bigl(u_1(t), \dots, u_n(t)\bigr),
\]
and the vector of growth modulation functions
\[
g(u,t) = \begin{bmatrix} g_1(u,t) \\ \vdots \\ g_n(u,t) \end{bmatrix}.
\]
The transition rates $k_{ij} \ge 0$ are collected into the $n \times n$ matrix $A = [a_{ij}]$ defined by
\begin{equation}\label{eq:defA}
    a_{ij} = 
    \begin{cases}
        k_{ji}, & i \neq j, \\[4pt]
        -\displaystyle\sum_{l \neq i} k_{il}, & i = j.
    \end{cases}
\end{equation}
Note that the off-diagonal indexes are transposed since the inflow into state $i$ from state $j$ is proportional to the abundance $u_j$. With these notations, system \eqref{eq:modelond} takes the compact form
\begin{equation}\label{eq:modelondMat}
    u' = R(t) \, \operatorname{diag}(u) \, g(u,t) + A u.
\end{equation}
The first term, $R(t) \operatorname{diag}(u) g(u,t)$, corresponds to the \emph{nonlinear vital dynamics} (proliferation, death, and competition), while the second, $A u$, corresponds to the \emph{linear transition dynamics}, describing how cells switch between states.

\medskip
\noindent\textbf{Relation to prior work.} Two-compartment models of the form \eqref{eq:modelo2d} are a workhorse of mathematical oncology, developed largely independently across applications. In breast cancer, Gavrilova et al.~\cite{gavrilova2026} model switching between HER2-positive and HER2-negative states to inform therapy sequencing. A similar two-compartment structure, simplifying previous approaches \cite{michor2005,roeder2006,getto2013}, describes quiescent versus cycling leukemic stem cells in chronic myeloid leukemia and, fitted to phase III trial data, predicts that many patients could be safely kept on a reduced tyrosine-kinase-inhibitor dose \cite{fassoni2018}. In adaptive therapy, Strobl et al.~\cite{strobl2021} use an analogous drug-sensitive/drug-resistant Lotka--Volterra model, fitted to longitudinal PSA data, to identify when competitive suppression of resistance is achievable; and in CAR-T therapy, Santurio et al.~\cite{santurio2024} use the same two-state architecture for antigen-positive and antigen-negative relapse. Models with more compartments represent the hierarchy of hematopoiesis \cite{marciniakczochra2009} and its disruption in chronic myeloid leukemia \cite{michor2005,roeder2006}. Here, we place these examples as a starting point within a common architecture, whose asymptotic and continuum behavior can be studied independently of the biological application.

\subsection{Models with uniform competition}
\label{sec:uniform_competition}

The main structural feature of model \eqref{eq:modelond} is the separation between vital and transition dynamics. As we show next, in a class of models with uniform competition, the long-term composition of the population depends on the latter, and, perhaps surprisingly, not on the former. This result is illustrated by the gray-shaded reaction terms in Figure~\ref{fig:roadmap}a,b.

\medskip
\noindent\textbf{Uniform versus non-uniform competition.}
In general, the modulation functions $g_i(u,t)$ may depend on the full population vector $u$ and differ among compartments. The most prominent example is the competitive Lotka--Volterra model,
\begin{equation}\label{eq:glv}
    g_i(u,t) = 1 - \frac{1}{K}\sum_{j=1}^n c_{ij} u_j,
\end{equation}
where the coefficients $c_{ij} \ge 0$ quantify the competitive pressure exerted by cells in state $j$ on cells in state $i$ \cite{zeeman1993}. In such \emph{non-uniform competition} models, the long-term behavior can depend strongly on the proliferation rates $r_i$ and on the interaction matrix $(c_{ij})$, leading to multiple attractors representing extinction, coexistence, bistability, etc.\ \cite{zeeman1993,zeemanvandendriessche1998,fassoni2019resilience}.

For many biological scenarios, however, especially phenotypic plasticity in cancer, it is plausible that all cell types compete for the same limited resources: although the compartments may proliferate at different rates, the brake on proliferation is the same for everyone and depends only on how crowded the tissue is. We call this the \emph{uniform competition} hypothesis. Under it, the long-term composition of the population does not depend on the proliferation rates $r_i$ and is determined by the transition matrix $A$ alone.

\medskip
\noindent\textbf{Asymptotic behavior under uniform competition.}
We consider model \eqref{eq:modelond} under four assumptions.

\begin{description}[leftmargin=2.2em, itemsep=3pt]
    \item[(H1) Common modulation.] All compartments feel the same brake on proliferation, which depends only on the total population $U=\sum_i u_i$: $g_i(u,t)=g(U)$ for all $i$.
    \item[(H2) Stable carrying capacity.] $g$ is continuously differentiable and has a unique zero $U^*>0$, with $g>0$ for $0<U<U^*$ and $g<0$ for $U>U^*$.
    \item[(H3) Connected plasticity.] Every state can be reached from every other state through a sequence of transitions with positive rates, i.e., the directed graph with an edge $i\to j$ whenever $k_{ij}>0$ is strongly connected.
    \item[(H4) Bounded and persistent proliferation.] The rates $r_i(t)$ are piecewise continuous with $0\le r_i(t)\le r_{\max}$, and, from some time $t_r\ge0$ on, at least one state $s$ proliferates persistently: $r_s(t)\ge r_{\min}>0$ for $t\ge t_r$.
\end{description}

Hypotheses (H1)--(H2) cover the logistic, generalized logistic, Gompertz and von Bertalanffy laws, $g(U)=1-U/K$, $1-(U/K)^\nu$, $\ln(K/U)$, $aU^{\gamma-1}-b$, among others \cite{benzekry2014classical, giaimo2025}; some of these also emerge from microscopic principles as manifestations of local contact inhibition \cite{universal2026contact}. Hypothesis (H3) expresses reversible plasticity: there are no absorbing states and no isolated groups of states. It holds, for instance, for sensitive/resistant or antigen-positive/antigen-negative switching, but fails for a strict differentiation hierarchy without any de-differentiation. The proof only requires that switching forgets its initial state. Thus, (H3) can be weakened \cite{companion2026}: it is sufficient that there is a single closed class of states reachable from every state and that the persistently proliferating state of (H4) belongs to it. The other states are transient, and the limit vanishes on them. Hypothesis (H4) allows quiescent states ($r_i=0$, as for quiescent leukemic stem cells) and time-dependent rates. It only requires that one compartment continue to proliferate from a certain point onward. Both (H3) and (H4) are needed: without switching ($A=0$), $u_1/u_2=\bigl(u_1(0)/u_2(0)\bigr)\exp\bigl((r_1-r_2)\int_0^tg(U)\,ds\bigr)$ retains the fitness difference and the initial data forever, and if all $r_i$ vanish after some time, the total population freezes before reaching $U^*$.

Under (H1)--(H4), the full model \eqref{eq:modelondMat} reads in matrix form as
\begin{equation}\label{eq:modelondMat2}
    u' = g(U) R(t) u + A u.
\end{equation}
Under (H3), the switching dynamics alone, given by
\begin{equation}
    u'=Au
    \label{eq:modelondg0}
\end{equation}
is the master equation of a continuous-time Markov chain \cite{disneyclarke1985} that is irreducible, with transition rate matrix $A$. The columns of $A$ sum to zero (mass conservation), and zero is a simple eigenvalue with left eigenvector $\mathbf{1}$ and right eigenvector $\pi$, while all other eigenvalues have negative real parts. Thus the linear system has a unique stationary distribution $\pi$ ($A\pi=0$, $\sum_i\pi_i=1$) with all entries positive. Moreover, it forgets its initial condition exponentially fast: there is a constant $\lambda>0$, depending only on the transition rates $k_{ij}$, such that solutions of \eqref{eq:modelondg0} approach $\pi$ (up to their total mass) at rate $e^{-\lambda t}$ (Lemma~\ref{lem:mix} in Appendix~A). Biologically, this means that plasticity erases lineage memory. A cell that starts in phenotype $i$ and switches according to the rates $k_{ij}$ will, after a time long compared with $1/\lambda$, be found in state $j$ with probability $\pi_j$, regardless of $i$; equivalently, a population with any initial composition relaxes to the composition $\pi$. Thus $1/\lambda$ is the \emph{mixing time} of the plasticity network. It is short when switching is fast and long when some transitions between groups of states are rare.

Theorem \ref{thm:disc} states that under (H1)--(H4), the full model \eqref{eq:modelondMat2} is asymptotically equivalent to the linear switching dynamics \eqref{eq:modelondg0}.

\begin{theorem}\label{thm:disc}
Assume (H1)--(H4) and let $u(t)$ be a solution of \eqref{eq:modelondMat2} with $u_i(0)\ge0$ and $U(0)>0$. Then:
\begin{enumerate}[label=(\alph*), itemsep=2pt, topsep=2pt]
\item the total population $U(t)$ is monotone and converges to the carrying capacity $U^*$;
\item the composition converges to the stationary distribution of the switching dynamics: $u(t)\to U^*\pi$;
\item consequently, the full system \eqref{eq:modelondMat2} and the purely linear transition system \eqref{eq:modelondg0} are asymptotically equivalent: $|u(t)-u^*(t)|\to0$ for \emph{every} solution $u^*$ of \eqref{eq:modelondg0} with total population $U^*$.
\end{enumerate}
\end{theorem}

Thus, in uniformly competitive models the long-term distribution of cells among compartments is governed only by the transition rates $k_{ij}$, i.e., by the plasticity architecture. The proliferation rates $r_i(t)$, and hence the intrinsic fitness of each phenotype, do not enter the final composition, and $g$ enters only through $U^*$. Plasticity, and not differences in fitness, sets the long-term heterogeneity of the population. The proof of Theorem~\ref{thm:disc} is given in Appendix~A.

\medskip
\noindent\textbf{Fitness leaves only a transient mark.} Theorem~\ref{thm:disc} rests on two ideas. First, at saturation all phenotypes are selectively neutral: the net growth rate of state $i$ is $r_i\,g(U)$, which vanishes at $U=U^*$ whatever the $r_i$, so the composition is shaped by transitions alone, as in the neutral theory of molecular evolution \cite{kimura1983}. Second, selection can act only while the population is away from its carrying capacity, and the mark it leaves is erased by switching. A population near $U^*$ will be barely affected by fitness differences, while a large expansion, such as the regrowth of a tumor after therapy, can shift the composition toward the fastest-proliferating states. This shift, however, then fades within the mixing time $1/\lambda$, which depends only on the transition rates $k_{ij}$ (Appendix~A). Therefore, the statement ``fitness is irrelevant'' is asymptotic. When switching is slow, as discussed in Section~\ref{sec:gradient}, the transient can be biologically long because the value of $\lambda$ is small.

\medskip
\noindent\textbf{Treatments.} The two-compartment model \eqref{eq:modelo2d} contains a treatment term $-d_2(t)u_2$ that acts on one compartment only, which violates (H1): while the drug is present, fitness differences do shape the composition, and this is precisely the regime where therapy is designed to act. Treatments, however, are applied for a finite period. Suppose that all treatment-induced losses $d_i(t)\ge0$ are bounded and vanish for $t\ge t_F$. From $t_F$ on, the system is again of the form \eqref{eq:modelondMat}, with initial condition $u(t_F)$, and Theorem~\ref{thm:disc} applies to it (hypothesis (H4) concerns large times, so it suffices that it hold from some time after $t_F$). Therefore, $u(t)\to U^*\pi$, exactly as if there had been no treatment. For this reason, a transient treatment reshapes the transient, and with it the relapse dynamics, but it cannot change the asymptotic composition. Lasting effects require changing the system itself, i.e., the transition rates $k_{ij}$ (for instance by epigenetic reprogramming). As above, how long the system takes to return to $U^*\pi$ is set by the mixing time $1/\lambda$, which is what matters clinically, and the deeper the reduction, the larger the imprint that fitness differences can leave on the regrowing population.

\medskip
\noindent\textbf{How far can uniform competition be relaxed?} The proof of Theorem~\ref{thm:disc} relies on (H1), so it is natural to ask whether the result survives when competition is only \emph{weighted}. Consider the Lotka--Volterra model \eqref{eq:glv} with $c_{ij}=c_j$, i.e., the competitive pressure is the same for all recipients $i$, but depends on the phenotype $j$ of the emitter. Setting $w_j=c_j/K$ and $W=\sum_jw_ju_j$, all compartments still share a common brake, $g_i(u,t)=g(W)$ with $g(W)=1-W$, but it now depends on the weighted population $W$ rather than on the total $U$. Switching conserves $U$ but not $W$, because a cell that changes phenotype changes how much it contributes to competition, and the mechanism behind the theorem no longer applies. In fact, in Appendix~A we give a counterexample in which the population oscillates indefinitely along a stable limit cycle. Uniform competition therefore cannot be relaxed to weighted competition.

\medskip
\noindent\textbf{Relation to prior work.} The idea that transitions can dominate over fitness in structured populations is not new. It is present in the neutral theory of molecular evolution \cite{kimura1983}, whose neutrality appears here at saturation, and in the view of phenotypic switching as a bet-hedging strategy in fluctuating environments \cite{kussellleibler2005}. Theorem~\ref{thm:disc} gives a precise version of this intuition for reversible phenotypic switching, under a uniform modulation of net growth. Fitness returns when this structure is lost. If phenotypes differ in birth rates and share a density-dependent death rate, the equilibrium composition depends on transitions and births \cite{zhou2013,niu2015}. Quasispecies theory belongs to this regime, since mutation is coupled to replication and the equilibrium depends jointly on fitness and mutation \cite{eigen1971}, as in the ``survival of the flattest'' of mutationally robust genotypes \cite{wilke2001}. Fitness also returns when switching is coupled to division, as when epigenetic changes arise from errors in copying DNA methylation at replication \cite{ushijima2003,yang2016clock}, because transition rates then scale with the division rate. Finally, the classical theory of competitive Lotka--Volterra systems, including the carrying simplex \cite{zeeman1993,zeemanvandendriessche1998} and basins of attraction and resilience under bistability \cite{fassoni2019resilience}, describes the fitness-dependent behavior that appears once (H1)--(H4) are relaxed.

On the experimental side, Gupta et al.~\cite{gupta2011} demonstrated phenotypic equilibrium in cancer cell lines and explained it by a Markov model of stochastic state transitions. Subsequent models combined cell-state conversions with proliferation, identifying the equilibrium of a growing population and establishing its stability at the level of averages and of individual trajectories \cite{zhou2013,niu2015}, and phenotypic switching has been analyzed as a driver of non-genetic drug resistance \cite{gunnarsson2020}. Theorem~\ref{thm:disc} gives conditions (uniform competition) under which, in contrast, the equilibrium composition of proliferating and competing cells is the stationary distribution of switching alone.

The closest mathematical antecedent to Theorem~\ref{thm:disc} is the work of Giaimo et al.~\cite{giaimo2025}, with the same motivation: to explain why linear Markov models of phenotypic switching predict the observed frequencies of cancer cell types despite nonlinear growth. They study the same model and prove asymptotic equivalence with the linear system through Yakubovich's theorem on perturbed linear systems. Our proof is more elementary, using only the variation-of-constants formula and a basic property of Markov chains, and it extends to the continuum (Theorem~\ref{thm:cont}; \cite{companion2026}). The hypotheses and conclusions also differ. (i)~For the growth law, they require an auxiliary function $h$ satisfying a specific differential condition involving $g$; our proof does not use such a particular assumption. (ii)~They require the abundance-weighted mean growth rate to stay above some $r_{\min}>0$, a condition on the solution; we ask that one compartment proliferate persistently, a condition on the parameters. (iii)~We additionally require connected plasticity (H3), which identifies the limit: without it, asymptotic equivalence still holds, but the shadowing linear solution depends on the $r_i$. (iv)~Asymptotic equivalence yields no rate, as Giaimo et al.\ note; our proof provides estimates that bound the distance to $U^*\pi$ in terms of the mixing rate $\lambda$, so that slow switching means a long memory of fitness differences. Finally, Giaimo et al.\ remark that it is unclear how Lotka--Volterra-type interactions could be included while preserving asymptotic equivalence; the counterexample of Appendix~A shows that, for weighted competition, asymptotic equivalence can indeed fail.

\subsection{From discrete compartments to a continuous phenotypic landscape}
\label{sec:continuous}

We now restrict the general model to transitions between neighboring states and show how the resulting model can be viewed as a finite-difference approximation to a partial differential equation on a continuous phenotypic space (the transition from panel~a to panel~c of Figure~\ref{fig:roadmap}). 

\medskip
\noindent\textbf{Compartmental model with local transitions.}
The general model \eqref{eq:modelond} allows transitions between any pair of compartments, including jumps between ``distant'' states. For purely phenotypic or epigenetic changes, however, cells typically change their state gradually (jumps are reintroduced in Section~\ref{sec:nonlocal}). We therefore consider a chain of $n$ compartments in which cells can switch only to immediately adjacent states:
\begin{equation}\label{eq:modelondloc}
    u_i' = r_i(t) u_i g_i(u,t) + k_{i-1,i} u_{i-1} + k_{i+1,i} u_{i+1} - (k_{i,i-1} + k_{i,i+1}) u_i,
\end{equation}
with the natural no-flux boundary conditions $k_{0,1}=k_{1,0}=k_{n,n+1}=k_{n+1,n}=0$.

Since \eqref{eq:modelondloc} is a particular case of \eqref{eq:modelond}, Theorem~\ref{thm:disc} still applies (a chain with positive rates in both directions satisfies the connectivity hypothesis (H3)): under the uniform competition hypotheses the system is asymptotically equivalent to the linear, transition-only counterpart $u' = A u$. The transition matrix $A$ is now \emph{tridiagonal}, reflecting nearest-neighbor coupling. In probabilistic terms, this linear system is the master equation of a continuous-time random walk on the finite state space $\{1, \dots, n\}$ with nearest-neighbor jumps: a cell hops from $i$ to $i+1$ at rate $k_{i,i+1}$ and to $i-1$ at rate $k_{i,i-1}$.

\medskip
\noindent\textbf{Connection to a reaction--diffusion--advection PDE.}
A natural question is whether the chain \eqref{eq:modelondloc} approximates a continuous description. We answer it by taking the reverse route: starting from a continuous reaction--diffusion--advection equation, we discretize it and show that the result coincides with \eqref{eq:modelondloc}. This also reveals how the transition rates $k_{i,j}$ relate to the diffusion and advection coefficients.

Consider a continuous phenotype coordinate $x \in \mathbb{R}$ and a population density $u(x,t)$ that evolves according to the reaction--diffusion--advection equation
\begin{equation}\label{eq:modeloidloc}
    \partial_t u = r(x,t)\, u \, g(u, x, t)  - \partial_x\!\bigl(v(x,t)\,u\bigr) + \partial_x\!\bigl(D(x,t)\,\partial_x u\bigr),
\end{equation}
where $v(x,t)$ is the advection velocity, $D(x,t) \ge 0$ the diffusivity, and the reaction term $r(x,t)\,u\,g(u,x,t)$ encodes the vital dynamics.

As detailed in Appendix~B, discretizing \eqref{eq:modeloidloc} on a uniform grid with step-size $\Delta x$ via the method of lines produces an ODE system with exactly the same structure as \eqref{eq:modelondloc}. Equating the coefficients of \eqref{eq:modelondloc} and the discretization of \eqref{eq:modeloidloc} leads to the following relations between the transition rates $k_{i,j}$ and diffusion and advection coefficients at the interface points $x_{i+1/2}$ between $x_{i}$ and $x_{i+1}$:
\begin{equation}\label{eq:escolhaDv}
        D(x_{i+1/2}) \approx \frac{\Delta x^2}{2} \bigl( k_{i,i+1} + k_{i+1,i} \bigr), \quad \quad
        v(x_{i+1/2}) \approx \Delta x \bigl( k_{i,i+1} - k_{i+1,i} \bigr).
\end{equation}

These relations have a direct interpretation. At the interface between neighboring states, the \emph{diffusion coefficient} $D$ is proportional to the \emph{average} of the forward and backward transition rates. Diffusion thus reflects the unbiased, random component of phenotypic switching; even if cells have no preferred direction, they will spread along the phenotype axis due to stochastic fluctuations. On the other hand, the \emph{advection velocity} $v$ at the interface is proportional to the \emph{difference} between the forward and backward rates. A nonzero $v$ indicates a directional bias: cells tend to move toward one end of the phenotypic axis more than the other. This bias can be interpreted as a deterministic driving force, such as a differentiation pressure or a gradient of an epigenetic potential, an interpretation that we explore in Section~\ref{sec:gradient}.

\medskip
\noindent\textbf{Boundary conditions.}
Boundary conditions must be specified for the continuous model. If the phenotype domain is bounded, $\Omega = [x_{\min}, x_{\max}]$, the natural choice consistent with the no-flux boundary conditions of the discrete model is the zero-flux (Robin-type) condition,
\[
    \bigl( v u - D \partial_x u \bigr)\big|_{x = x_{\min}} = 0, \qquad
    \bigl( v u - D \partial_x u \bigr)\big|_{x = x_{\max}} = 0,
\]
which ensure that no cells enter or leave the system through the boundaries of the phenotype space. For an unbounded domain $\Omega = \mathbb{R}$, biologically reasonable solutions should satisfy $u(x,t) \to 0$ and $\partial_x u(x,t) \to 0$ as $|x| \to \infty$.

\medskip
\noindent\textbf{Matrix form: symmetric and skew-symmetric parts.}
The same identification can be read off the transition matrix. The tridiagonal matrix $A$ can be uniquely decomposed into its symmetric and skew-symmetric parts, $A = A_{\text{sym}} + A_{\text{skew}}$. Using the identifications \eqref{eq:escolhaDv}, we find that the symmetric part carries the diffusion and the skew-symmetric part the advection: their off-diagonal entries are the interface values of $D$ and $v$. The compartmental system can then be rewritten as (see Appendix~B for full details)
\begin{equation}\label{eq:matrixDecomp}
    u' = r(x,t)\, u \, g(u, x, t) + \frac{1}{\Delta x^2} \, \mathcal{D} \, u + \frac{1}{2\Delta x} \, \mathcal{V} \, u,
\end{equation}
where $\mathcal{D}=\Delta x^2A_{\text{sym}}$ and $\mathcal{V}=2\Delta x\,A_{\text{skew}}$ are matrix approximations of the diffusive and advective operators, respectively (when $v$ varies, the diagonal of the symmetric part also contains an advective term; Appendix~B gives the precise correspondence). A single matrix of transition rates thus contains two mechanisms: \emph{unbiased random motion}, in its symmetric (diffusive) part, and \emph{directed motion}, in its skew-symmetric (advective) part.

\medskip
\noindent\textbf{Relation to prior work.} The passage from discrete compartmental models to continuous reaction--diffusion--advection equations is standard. On the discrete side, nearest-neighbor transitions of the kind in \eqref{eq:modelondloc} appear directly in models of hierarchically organized tissues, such as the invasion of de-differentiating cancer cells across differentiation levels \cite{zhoutraulsen2019}; but coupling this discrete structure to phenotypic switching, with an explicit identification of the transport coefficients from the microscopic rates, is less common. Lorenzi et al.~\cite{villa2025} provide a tutorial on the resulting phenotype-structured PDEs (PS-PDEs), covering their derivation from individual-based models, the analysis of traveling waves and concentration phenomena, and numerical methods including the method-of-lines discretization we also employ; their derivation proceeds from a discrete stochastic model (a branching random walk on a lattice) via a Taylor expansion of the transition probabilities, a route parallel in spirit to the finite-difference identification we carry out in Appendix~B, though starting from individual cell rules rather than from a population-level compartmental model. Agostinelli et al.~\cite{agostinelli2026} develop an alternative, systematic derivation of continuum limits from discrete structured population models based on matched asymptotic expansions and multiple scales, which handles boundary layers and regions where a continuum description breaks down more rigorously than the direct Taylor-expansion route we use here. In a closely related direction, Alvarez et al.~\cite{clairambault2022} study integro-differential equations for populations endowed with plasticity of traits; they interpret the diffusion term in trait space as ``non-genetic instability," a continuous analogue of epigenetic mutations, and derive advection--diffusion equations as local approximations via a Kramers--Moyal expansion, an identification of $D(x)$ and $v(x)$ consistent with the one we obtain here from the symmetric and skew-symmetric parts of the transition matrix (equation~\eqref{eq:escolhaDv}). Phenotype-structured integro-differential and PDE models with non-genetic instability (random epimutations) and selection, applied to drug tolerance and relapse, were developed in \cite{chisholm2015,lorenzi2015,lorenzi2016}, within the structured-population framework of Perthame~\cite{perthame2007}; see \cite{clairambaultpouchol2019} for a survey of such models of drug resistance. Continuous models of phenotypic dynamics have also been applied directly to specific biological systems, notably hematopoiesis and leukemia: Cho et al.~\cite{cho2018} model acute myeloid leukemia as a continuum of differentiation states using diffusion maps constructed from single-cell data and infer the geometry of the phenotypic manifold from data, and Singh et al.~\cite{singh2025population} show, with a population-dynamics model of hematopoiesis informed by single-cell data, that myeloid bias involves both stem-cell differentiation and progenitor proliferation biases. Integro-differential Lotka--Volterra models with phenotype-dependent reproduction rates have also been analyzed together with optimal-control problems \cite{de2024analysis}.

\subsection{Fitness-independent asymptotic distributions in the continuum}
\label{sec:pdeasymp}

Since the discrete model is an approximation of the continuous one (Section~\ref{sec:continuous}), and fitness is asymptotically irrelevant in discrete models under uniform competition (Theorem~\ref{thm:disc}), it is natural to ask whether the same holds in the continuum: is the long-term phenotypic distribution still governed solely by the diffusion and advection terms, with the nonlinear vital dynamics becoming irrelevant? The answer is yes, under the continuum counterparts of (H1)--(H4).

\medskip
\noindent\textbf{Continuum hypotheses.}
Let the phenotype space be a bounded interval $\Omega=(x_{\min},x_{\max})$ with the no-flux boundary conditions of Section~\ref{sec:continuous}, and let $U(t)=\int_\Omega u(x,t)\,dx$ be the total population. We assume:
\begin{description}[leftmargin=2.2em, itemsep=3pt]
    \item[(H1$_{\text{c}}$) Common modulation.] $g(u,x,t)=g(U(t))$ for all $x\in\Omega$.
    \item[(H2$_{\text{c}}$) Stable carrying capacity.] $g$ is continuously differentiable and has a unique zero $U^*>0$, with $g>0$ for $0<U<U^*$ and $g<0$ for $U>U^*$.
    \item[(H3$_{\text{c}}$) Non-degenerate switching.] $D$ is twice and $v$ once continuously differentiable, and $D(x)>0$ on $[x_{\min},x_{\max}]$.
    \item[(H4$_{\text{c}}$) Bounded and persistent proliferation.] $0\le r(x,t)\le r_{\max}$, and from some time on $r(x,t)\ge r_{\min}>0$ on some sub-interval of phenotypes of positive length.
\end{description}
Hypothesis (H3$_{\text{c}}$) is the continuum counterpart of connected plasticity: diffusion acts everywhere, so every phenotype can be reached from every other.

\medskip
\noindent\textbf{Stationary distribution of the switching dynamics.}
In the discrete model the asymptotic composition is the stationary distribution $\pi$ of the Markov chain. Its continuum analogue is the steady solution of the purely advective--diffusive equation
\begin{equation}\label{eq:modeloidloclin}
    \partial_t u = -\partial_x(v u) + \partial_x(D\,\partial_x u).
\end{equation}
Setting $\partial_t u=0$ and integrating once gives $-vu+D\,\partial_x u=J$ (constant flux); the no-flux conditions force $J=0$, so $D\,\partial_x u=v\,u$, and
\begin{equation}\label{eq:stationary}
    u_{\text{eq}}(x) = U^*\,\psi(x), \qquad \psi(x) = \frac{1}{Z}\exp\!\Bigl( \int_{x_0}^x \frac{v(s)}{D(s)}\,ds \Bigr),
\end{equation}
where $Z$ normalizes $\psi$ to a probability density.

\begin{theorem}[informal]\label{thm:cont}
Assume (H1$_{\text{c}}$)--(H4$_{\text{c}}$). Then, for every nonnegative initial density with $U(0)>0$, the total population converges monotonically to $U^*$ and the density converges to $u_{\mathrm{eq}}=U^*\psi$ in $L^1(\Omega)$, whatever the proliferation rates $r(x,t)$; the full model \eqref{eq:modeloidloc} is asymptotically equivalent to the linear equation \eqref{eq:modeloidloclin}; and an estimate of the same form as \eqref{eq:star} holds, with the rate $\lambda$ given by the spectral gap of the switching dynamics.
\end{theorem}

The proof is beyond the scope of this work and is given in \cite{companion2026}. It generalizes the proof of Theorem~\ref{thm:disc} in Appendix~A step by step. The same reference extends the result to the whole real line, under two natural conditions (the landscape confines the population, and cells cannot escape to infinitely distant phenotypes), which cover the Ornstein--Uhlenbeck, multi-well and heavy-tailed examples of Section~\ref{sec:examples}. It also covers bounded phenotype domains in higher dimension, with drifts that need not be gradients, and switching by jumps, alone or combined with drift and diffusion as in the unified model of Section~\ref{sec:nonlocal}.

Figure~\ref{fig:imprint} illustrates Theorem~\ref{thm:cont} and isolates the imprint of fitness. Anticipating Section~\ref{sec:gradient}, we take the velocity to be the negative gradient of a potential, $v=-P'$, whose wells confine the population; here $P$ is a tilted double well (Figure~\ref{fig:imprint}a), and proliferation strongly favors the right well (Figure~\ref{fig:imprint}b). The population starts at the stationary composition, $u(x,0)=0.01\,\psi(x)$, so that every deviation from $\psi$ is caused by fitness differences. During the growth phase, the composition is pulled towards the fast-proliferating well; once the population has saturated, switching brings it back to $\psi$ (Figure~\ref{fig:imprint}c), and the imprint of fitness decays at the rate $\lambda_1$ given by the spectral gap of the switching dynamics (Figure~\ref{fig:imprint}e). The mixing time $1/\lambda_1$ grows exponentially with the ratio between barrier height and noise intensity (Figure~\ref{fig:imprint}f): in this example it ranges from about one to about $900$ time units, while growth lasts about five.

\begin{figure}[htbp]
\centering
\includegraphics[width=0.98\textwidth]{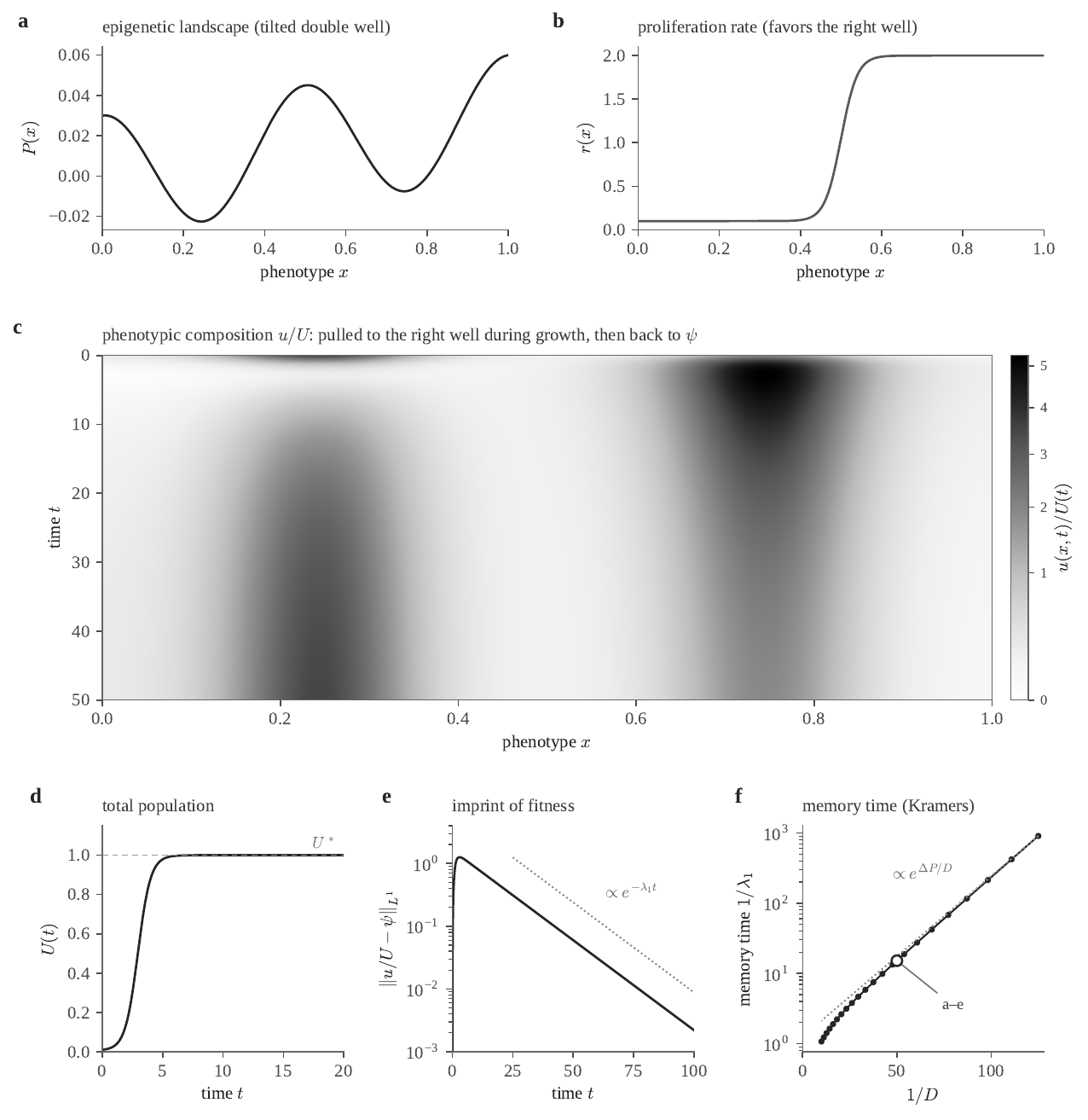}
\caption{\footnotesize\linespread{1}\selectfont\textbf{Illustration of Theorem~\ref{thm:cont}: uniform competition erases the imprint of fitness.} The full reaction--advection--diffusion equation \eqref{eq:modeloidloc} with $g(U)=1-U$ and $v=-P'$ on $(0,1)$, with zero-flux boundaries and constant diffusivity $D=0.02$. \textbf{(a)} The potential, a tilted double well, $P(x)=0.03\cos(4\pi x)+0.03\,x$; the stationary density $\psi\propto e^{-P/D}$ places $68\%$ of the population in the left well. \textbf{(b)} The proliferation rate, $r(x)=0.1+1.9/\bigl(1+e^{-(x-1/2)/0.02}\bigr)$, which strongly favors the right well. \textbf{(c)} Kymograph of the composition $u/U$ (color scale nonlinear, power law with exponent $0.6$; time increases downward). The population starts at the stationary composition, $u(x,0)=0.01\,\psi(x)$; during growth it is pulled towards the right well, and afterwards switching brings it back to $\psi$. \textbf{(d)} Total population $U(t)$, which reaches the carrying capacity $U^*=1$ by $t\approx5$. \textbf{(e)} $L^1$ distance of the composition to $\psi$: it peaks at $1.26$ during growth and then decays at the spectral gap $\lambda_1\approx0.066$ of the switching dynamics (dotted: slope $-\lambda_1$). \textbf{(f)} Memory time $1/\lambda_1$ as a function of $1/D$, for the same potential: it grows like $e^{\Delta P/D}$ (dotted), where $\Delta P\approx0.053$ is the barrier seen from the shallower well (the Kramers regime); the circle marks $D=0.02$, used in (a)--(e). Finite-volume discretization with $N=400$ cells whose fluxes preserve $\psi$ exactly, integrated in time with a stiff (BDF) solver.}
\label{fig:imprint}
\end{figure}

The biological message is the one of Section~\ref{sec:uniform_competition}: \emph{under uniform competition, the long-term phenotypic distribution is independent of the proliferation rates $r(x,t)$ and is determined by the ratio $v/D$ alone}. A phenotype with high fitness but strong outward transitions does not dominate, and a low-fitness phenotype can persist if inward fluxes are large enough. The same qualifications also carry over: the statement is asymptotic, and the relevant time scale is the inverse spectral gap of the switching dynamics, which, for a landscape with deep wells separated by high barriers, can be exponentially long in the barrier height (the Kramers regime; Figure~\ref{fig:imprint}f). The shape of the stationary distribution (uniform, exponential, Gaussian, or heavy-tailed) emerges from the balance between deterministic drift and random diffusion, as we explore in Section~\ref{sec:examples}, and a second numerical illustration of Theorem~\ref{thm:cont}, on the multi-well landscape of Figure~\ref{fig:roadmap}e,f, is given in Section~\ref{sec:gradient} (Figure~\ref{fig:pdesim}).

\medskip
\noindent\textbf{Relation to prior work.} Theorem~\ref{thm:cont} connects the continuum limit of Section~\ref{sec:continuous} with the asymptotic equivalence result of Section~\ref{sec:uniform_competition}, and so links the compartmental and the PS-PDE literatures. The general framework of Lorenzi et al.~\cite{villa2025} and the integro-differential formulations of \cite{clairambault2022,chisholm2015,lorenzi2016} retain the nonlinear reaction term throughout, and we have not found the question of its asymptotic negligibility formulated there. One reason may be that many phenotype-structured models of adaptive dynamics use an additive form of competition, $r(x)-\kappa U$, rather than the multiplicative form $r(x)\,g(U)$ \cite{perthame2007}. With additive competition, selection acts at full strength for all times, the asymptotic phenotypic distribution depends on $r(x)$, and fitness is not forgotten, as illustrated numerically in \cite{companion2026}. In the limit of small phenotypic changes, such populations concentrate on the fittest traits \cite{perthame2007,lorz2011}, and the balance between epimutations and selection has been analyzed in fluctuating environments \cite{lorenzi2015}. Which form of density dependence applies to a given tumor therefore matters: it decides whether the long-term phenotypic distribution is set by plasticity alone or by plasticity and fitness jointly.

\subsection{The individual cell perspective: Fokker--Planck equation and connection to stochastic processes}
\label{sec:stochastic}

The linear advection--diffusion equation \eqref{eq:modeloidloclin} that governs the asymptotic population density under uniform competition,
\begin{equation}
\label{eq:eqlinssprob}
    \partial_t u = -\partial_x\bigl(v(x) u\bigr) + \partial_x\bigl(D(x) \partial_x u\bigr),
\end{equation}
has a well-known probabilistic interpretation. It is also the \emph{Fokker--Planck equation} (also called the forward Kolmogorov equation) for a continuous-time stochastic process describing the trajectory of a single cell in the phenotype space (Figure~\ref{fig:roadmap}e) \cite{risken1996}. We now make this connection explicit.

\medskip
\noindent\textbf{Langevin and Fokker--Planck equations.}
Consider a stochastic differential equation, or \emph{Langevin equation},
\begin{equation}\label{eq:general_sde}
    dX_t = a(X_t)\,dt + b(X_t)\,dW_t,
\end{equation}
for the position $X_t$ of a particle driven by a drift $a(x)$ and by a standard Brownian motion $W_t$ with noise amplitude $b(x)$; we interpret it, and every SDE in this paper, in the It\^o sense \cite{oksendal2003stochastic}. The probability density $p(x,t)$ of $X_t$, or, equivalently, the density of a cloud of independent particles each following \eqref{eq:general_sde}, obeys the associated Fokker--Planck equation
\begin{equation}\label{eq:general_FP_classical}
    \partial_t p = -\partial_x\bigl[ a\, p \bigr] + \tfrac{1}{2}\,\partial_x^2\bigl[ b^2\, p \bigr],
\end{equation}
which, written as a conservation law with explicit advective and diffusive fluxes, reads
\begin{equation}\label{eq:general_FP}
    \partial_t p = -\partial_x\Bigl[ \bigl( a - \tfrac{1}{2}(b^2)' \bigr)\, p \Bigr] + \partial_x\Bigl[ \tfrac{1}{2} b^2\, \partial_x p \Bigr].
\end{equation}

\medskip
\noindent\textbf{From population density to individual trajectories.}
Reading the advection--diffusion equation \eqref{eq:eqlinssprob} as a Fokker--Planck equation in the form \eqref{eq:general_FP}, matching the diffusive fluxes gives $b(x) = \sqrt{2D(x)}$, and matching the advective fluxes gives the drift $a(x) = v(x) + D'(x)$. Therefore, the macroscopic description, itself the continuum version of the compartmental model, is also the Fokker--Planck description of the following Langevin equation, in It\^o form, describing the individual cell trajectory $X_t$ along the phenotype space:
\begin{equation}\label{eq:SDEito}
    dX_t = \bigl( v(X_t) + D'(X_t) \bigr) \, dt + \sqrt{2 D(X_t)} \, dW_t.
\end{equation}
The macroscopic distribution of the population is thus the probability density of an ensemble of identical, independently moving cells, a duality between the population-level PDE and the individual-level SDE that is central to the mechanistic interpretation of the model. The extra term $D'$, which vanishes when $D$ is constant, is required by the It\^o convention for the Fokker--Planck equation of \eqref{eq:SDEito} to be the flux-form equation \eqref{eq:eqlinssprob}; its effect is discussed in Section~\ref{sec:gradient}. For the same Fokker--Planck equation, the drift would be $v+D'/2$ in the Stratonovich convention and $v$ in the anti-It\^o (H\"anggi--Klimontovich) convention; the convention only determines how the SDE is written, once the Fokker--Planck equation is fixed \cite{lau2007}.

\medskip
\noindent\textbf{The Ornstein--Uhlenbeck process: return to the mean with noise.}
An important special case arises when the diffusivity is constant, $D(x) \equiv D$, and the velocity is linear and restoring,
\begin{equation}\label{eq:OUvelocity}
    v(x) = -\theta (x - \mu),
\end{equation}
with $\theta > 0$ and a reference phenotype $\mu$. The SDE \eqref{eq:SDEito} then becomes
\begin{equation}\label{eq:OUSDE}
    dX_t = -\theta (X_t - \mu) \, dt + \sqrt{2D} \, dW_t,
\end{equation}
which is the classical \emph{Ornstein--Uhlenbeck} (OU) process. The drift term $-\theta (X_t - \mu)$ represents a restoring force that pulls the phenotype back toward a target phenotype $\mu$, a valley of the Waddington landscape: the farther $X_t$ is from $\mu$, the stronger the drift. The constant noise intensity $\sqrt{2D}$ adds persistent random fluctuations. Its stationary distribution is the Gaussian discussed in Section~\ref{sec:examples} (Case~3).

\medskip

\noindent\textbf{Relation to prior work.} The use of the Ornstein--Uhlenbeck (OU) process to model phenotypic evolution has a long history in quantitative genetics and biophysics. De Souza Silva et al.~\cite{desouzasilva2023} analyze phenotypic evolution as an OU process under environmental variation and phenotypic plasticity: under a fixed optimum the stationary distribution is Gaussian, as in Case~3 of Section~\ref{sec:examples}, but a moving optimum combined with plasticity deforms the effective potential into skewed, non-Gaussian distributions, with an upper bound on the rate of environmental change beyond which the population cannot persist. In the cancer context, Kessler and Levine~\cite{kessler2022} model a ``chance to persist" (CTP) phenotype with two coupled variables: a mutation--selection equation for the population-level CTP density, whose steady state is an Airy function rather than a Gaussian (from a reflecting boundary near the origin), and a genuine OU process with an absorbing boundary for an individual-level survival trait, from which they derive a closed-form persistence probability under therapy. A recent extension couples this picture to a drug-concentration-dependent reaction--diffusion--advection equation over a two-dimensional epigenetic--phenotypic space and predicts an optimal drug concentration and drug-holiday schedule \cite{parklevine2025}, the kind of concrete, testable prediction a mechanistic account of $D(x)$ and $v(x)$ should also be able to generate. In contrast to these phenomenological OU constructions, here the drift and diffusion coefficients are expressed through microscopic switching rates (Eq.~\eqref{eq:escolhaDv}) rather than posited directly, and the OU process appears as one member of a broader family of potentials (Section~\ref{sec:examples}).

\subsection{The velocity field as a gradient: effective potentials and stochastic gradient descent on the phenotypic landscape}
\label{sec:gradient}

So far, the advection velocity $v(x)$ has been identified with the net difference between the transition rates connecting neighboring compartments. From the perspective of the PDE itself, however, $v(x)$ is a phenomenological input: it is postulated \emph{a priori} as a velocity field that drives the motion, but, unlike diffusion, which emerges naturally from random switching between states, it is not derived from an underlying mechanism. We now argue that this velocity field admits a mechanistic reading, grounded in the tendency of cells to settle into stable configurations: $v(x)$ can be interpreted as the downhill direction of an epigenetic potential whose valleys are the stable phenotypes.

\medskip
\noindent\textbf{Stability and effective potentials in epigenetic regulation.}
Among the molecular mechanisms that drive epigenetic change, histone modification is one of the best studied \cite{bird2007,kouzarides2007,allisjenuwein2016}. Histones are the proteins around which DNA is wrapped, and their configuration determines which genes are accessible for transcription in a given cell. This configuration is, in turn, shaped by molecular forces that fold histones into stable conformations; chemical modifications such as acetylation and methylation alter these forces, driving histones toward new stable conformations and thereby changing which regions of DNA are exposed or occluded. At the level of single molecules and assemblies, stable conformations are minima of a free energy, as in the folding of proteins \cite{dillchan1997} and the assembly of nuclear compartments \cite{banani2017}. At the level of gene-expression programs, however, the relevant dynamics is driven out of equilibrium: histone marks are written and erased by enzymes that consume energy, and a stable chromatin state is a self-sustaining steady state of these reactions, maintained by the recruitment of modifying enzymes by existing marks, rather than a minimum of a thermodynamic free energy \cite{dodd2007,michieletto2016}. What survives at this level is the landscape in the sense of an effective potential, or quasi-potential, of the stochastic dynamics \cite{wang2015,zhou2016}: its local minima are metastable gene-expression programs, separated by barriers set by the cost of rewriting the histone code \cite{alarcon2026}. A cell type thus corresponds to a local minimum of this effective potential, and transitions between phenotypes are barrier-crossing events, driven by molecular noise and enzymatic activity. This is the sense in which we give Waddington's epigenetic landscape a physical reading.

\medskip
\noindent\textbf{From landscapes to velocity fields and stationary distributions.}
This landscape can be expressed in mathematical form as a function of the phenotype coordinate $x$ if we assume that the deterministic force $v(x)$ that drives a cell toward stable states is the negative gradient of an epigenetic potential $P(x)$,
\begin{equation}\label{eq:gradient}
    v(x) = -P'(x).
\end{equation}
Mathematically, any smooth velocity field can be written this way in one dimension ($P(x)=-\int_{x_0}^x v(s)ds$), but the physical hypothesis is that $P(x)$ describes a biologically meaningful epigenetic landscape (Figure~\ref{fig:roadmap}c,e).

Inserting $v=-P'$ into the continuous model \eqref{eq:modeloidloc} yields, under uniform competition (Theorem~\ref{thm:cont}) and in a tissue in homeostasis ($U= U^*$ and $g= 0$), the asymptotic equation
\begin{equation}\label{eq:FPgrad}
    \partial_t u = \partial_x\!\bigl(P'(x)\,u\bigr) + \partial_x\!\bigl(D(x)\,\partial_x u\bigr).
\end{equation}
Its stationary solution follows directly from \eqref{eq:stationary}:
\begin{equation}\label{eq:stationarypot}
    u_{\mathrm{eq}}(x) = C \exp\!\Bigl( -\int_{x_0}^x \frac{P'(s)}{D(s)}\,ds \Bigr).
\end{equation}
For constant diffusivity $D(x)\equiv D$, this simplifies to the Boltzmann distribution
\[
    u_{\mathrm{eq}}(x) = C \exp\!\bigl(-P(x)/D\bigr),
\]
where $D$ plays the role of an effective temperature: high noise flattens the distribution and facilitates transitions between minima.

When $D(x)$ varies, the stationary distribution can still be written in a Boltzmann form $$u_{\mathrm{eq}}(x) = C \exp\!\bigl(-\Phi(x)\bigr),$$ where the effective potential $\Phi(x)$ satisfies $\Phi'(x)=P'(x)/D(x)$. With state-dependent noise, the model can describe, for instance, faster exploration of the noisier regions of the landscape (Section~\ref{sec:examples}).

In more than one dimension, the gradient form \eqref{eq:gradient} is a genuine restriction: nonequilibrium drifts generically include a rotational part that breaks detailed balance and sustains circulating probability fluxes \cite{wang2015}, and the stationary density is then no longer determined by $P$ alone. The irrelevance of fitness does not depend on this assumption, since Theorem~\ref{thm:disc} holds for any irreducible transition matrix, reversible or not, and its continuum counterpart holds on bounded phenotype domains in any dimension, with drifts that need not be gradients \cite{companion2026}: the limit is then the stationary density of the switching dynamics, which is not of Boltzmann form and carries a circulating probability flux. Only the Boltzmann-type form of the limit depends on the gradient assumption. Landscape models of the epithelial--mesenchymal transition, in which epithelial, mesenchymal and hybrid states appear as attractors of a core regulatory circuit, are the most developed application of this kind in cancer \cite{lu2013emt,li2016emt}.

\medskip
\noindent\textbf{A numerical illustration.}
Figure~\ref{fig:pdesim} illustrates Theorem~\ref{thm:cont} on the multi-well landscape sketched in Figure~\ref{fig:roadmap}e,f. It shows the numerical solution of the full nonlinear reaction--advection--diffusion equation \eqref{eq:modeloidloc}, obtained with the method of lines. Cells are seeded, via a narrow Gaussian, at the left end of the domain, and proliferate at a rate $r(x)=r_0-mx$ that is highest where the population starts, favoring states near the initial, shallow minimum; reflecting (zero-flux) boundaries are imposed at both ends. The population grows rapidly to carrying capacity while beginning to redistribute across the landscape (Figure~\ref{fig:pdesim}c), and its profile converges to the Boltzmann-type stationary density~\eqref{eq:stationarypot} of the purely linear equation, with no visible trace of the fitness gradient that shaped the transient (Figure~\ref{fig:pdesim}e). Here the barriers are moderate, so the imprint of fitness fades quickly once the population has saturated; with higher barriers the same convergence would take much longer, as discussed in Section~\ref{sec:pdeasymp}. Because the population is seeded where proliferation is fastest, this simulation shows convergence from a distant initial condition rather than isolating the imprint of fitness, which is the purpose of Figure~\ref{fig:imprint}; the landscape is the one of Figure~\ref{fig:roadmap}e,f, and it allows a direct comparison with the single-cell simulation of Figure~\ref{fig:examples}e.

\begin{figure}[htbp]
\centering
\includegraphics[width=0.98\textwidth]{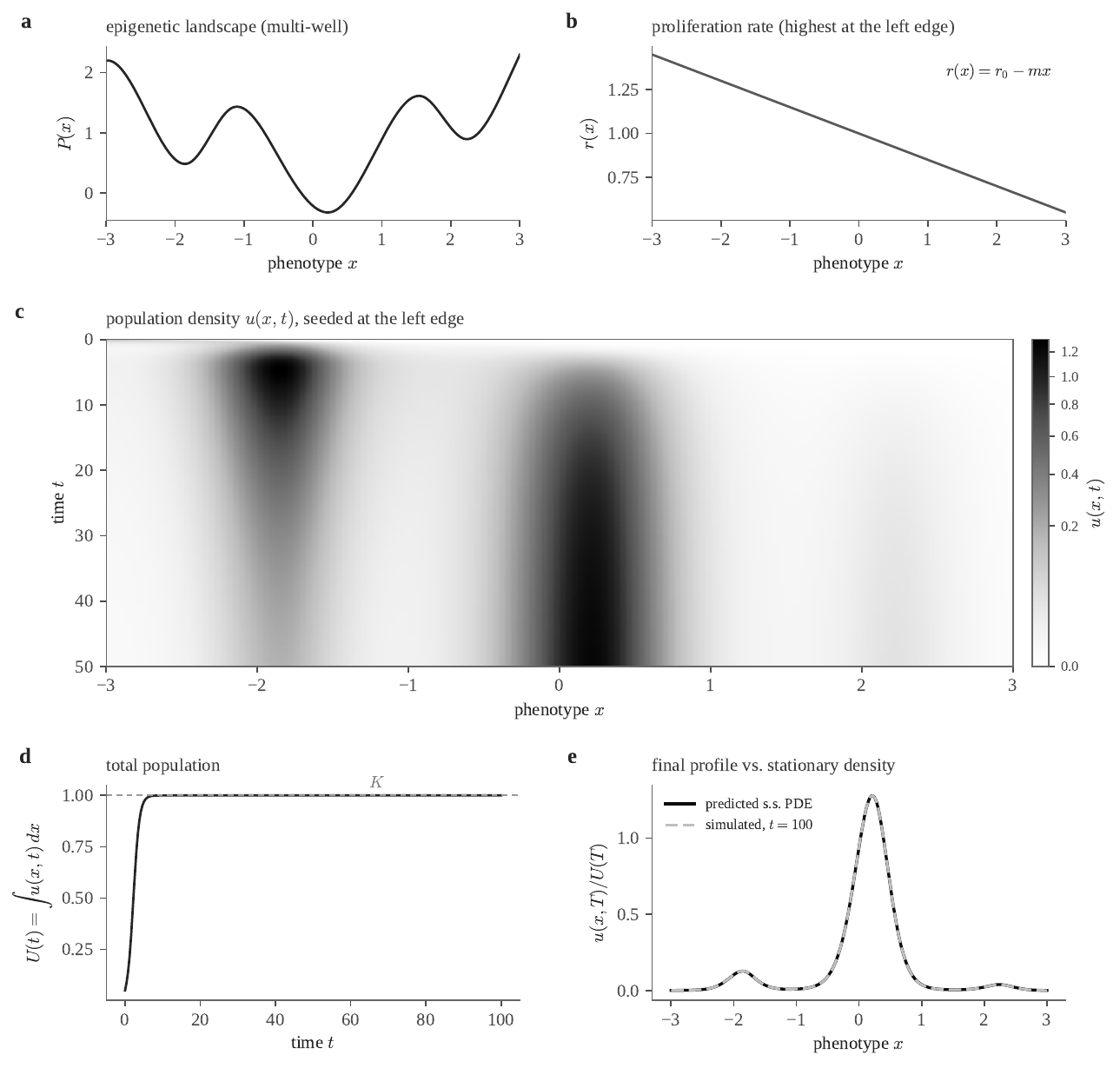}
\caption{\footnotesize\linespread{1}\selectfont\textbf{Illustration of Theorem~\ref{thm:cont}: the full nonlinear PDE converges to the stationary density of the switching dynamics despite a spatially decreasing fitness.} The full reaction--advection--diffusion equation \eqref{eq:modeloidloc} is solved directly by the method of lines. Space is discretized into $N=150$ compartments, using the same finite-difference identification of $D$ and $v$ from local transition rates as in Section~\ref{sec:continuous}; this turns the PDE into exactly the type of large compartmental ODE system covered by Theorem~\ref{thm:disc}, which is then integrated in time with a stiff ODE solver. \textbf{(a)} The potential $P(x)$, the same multi-well landscape as Figure~\ref{fig:roadmap}e,f (a cubic spline through fixed control points; see Data availability), with constant diffusivity $D=0.35$. \textbf{(b)} The proliferation rate $r(x) = r_0 - mx$ ($r_0=1$, $m=0.15$), decreasing across phenotype space; biologically, this represents a population in which the least differentiated phenotype (left edge) is also the most proliferative, as for stem and progenitor cells relative to their post-mitotic descendants. \textbf{(c)} Kymograph of $u(x,t)$ (color scale nonlinear, power-law with $\gamma=0.45$, to make the shallow right-hand well visible; time increases downward, as in Figure~\ref{fig:examples}). The population is seeded, at $t=0$, as a narrow Gaussian at the left boundary ($x_0=-3$, $\sigma=0.25$), with total initial mass $U(0)=0.05$, and reflecting (zero-flux) boundaries are imposed at $x=\pm3$. \textbf{(d)} Total population $U(t)$, growing rapidly to the carrying capacity $K=1$ under the logistic modulation $g(U)=1-U/K$; note that essentially all growth, and hence all activity of the reaction term $r(x)u\,g(U)$, occurs within the first $t\approx10$ time units. \textbf{(e)} The simulated profile at $t=100$, normalized to a probability density, against the stationary density $u_{\text{eq}}(x)\propto\exp(-P(x)/D)$ predicted by the purely linear equation \eqref{eq:FPgrad} (no reaction term): the two are visually indistinguishable ($L^1$ distance $2.6\times10^{-3}$), with no trace of the fitness advantage that favored the starting region throughout the transient.}
\label{fig:pdesim}
\end{figure}

\medskip
\noindent\textbf{Stochastic gradient descent on the epigenetic landscape.}
Under uniform competition, the long-term behavior of the cell density is governed by the linear PDE \eqref{eq:FPgrad} (Theorem~\ref{thm:cont}), the Fokker--Planck equation of the single-cell SDE \eqref{eq:SDEito}. We now read this SDE under the gradient hypothesis $v(x) = -P'(x)$. With this substitution, the SDE for the individual cell trajectory, equation \eqref{eq:SDEito}, becomes
\begin{equation}\label{eq:SDEgradient}
    dX_t = \bigl( -P'(X_t) + D'(X_t) \bigr) \, dt + \sqrt{2D(X_t)} \, dW_t.
\end{equation}
This SDE can be read as a \emph{stochastic gradient descent} (SGD) on the epigenetic landscape: the drift derives from the minimization of the potential $P$, while the noise stems from diffusion, i.e., from the random motion of cells between neighboring compartments. Each cell is driven in the direction of the negative gradient of $P$, corrected by the term $D'$ when noise is not uniform (see below), seeking a stable configuration; without noise, it would deterministically roll downhill to the nearest minimum, corresponding to a stable phenotype. Intrinsic noise and variability perturb this motion, however, so that convergence to an exact minimum is no longer guaranteed: fluctuations can carry a cell over the barriers separating minima, allowing it to explore the landscape and transition between basins of attraction. The expression \emph{stochastic gradient descent} is borrowed from optimization and machine learning \cite{robbinsmonro1951,bottoucurtisnocedal2018}. We use it in the Langevin sense, gradient descent with an explicit noise term, and not as a claim of equivalence with the training algorithms of machine learning, whose randomness comes from subsampling the data; we return to this connection in the Conclusion.

\medskip
\noindent\textbf{State-dependent noise and the effective potential.}
The drift in \eqref{eq:SDEgradient}, written in the It\^o convention, is not the bare $-P'$: it contains the term $D'(X_t)$, whose effect depends on whether noise is uniform across the landscape. When $D(x)\equiv D$ is constant, $D'(X_t)$ vanishes identically, and noise acts purely as an effective temperature that agitates trajectories without biasing them; the resulting stationary density is then a simple rescaling of the underlying potential itself, the Boltzmann form $u_{\mathrm{eq}}(x) = C\exp(-P(x)/D)$ derived above. When $D(x)$ varies in phenotype space, the term $D'$ pushes cells toward regions of \emph{higher} noise intensity. Its net effect is captured by the effective potential $\Phi(x)$, with $\Phi'=P'/D$, whose Boltzmann density $e^{-\Phi}$ is flatter where $D$ is large: at equal slope of $P$, barriers in noisier regions are effectively lower and cells explore them more freely. This is the sense in which noise reshapes the effective landscape (Section~\ref{sec:examples}).

The gradient and the noise thus play opposite roles: the first keeps cells in their phenotypes, the second lets them change.

\medskip
\noindent\textbf{Inferring the landscape from data.}
The gradient hypothesis also suggests measuring the landscape rather than postulating $v$ and $D$. Single-cell technologies (scRNA-seq, ATAC-seq, ChIP-seq) provide snapshots of the distribution of cells across a phenotypic manifold and, if the population is near steady state, \eqref{eq:stationarypot} can be inverted: $P'/D=-\partial_x\ln u_{\mathrm{eq}}$, that is, $P=-D\ln u_{\mathrm{eq}}$ up to a constant when $D$ is constant. The steady-state assumption matters: out of saturation, or under non-uniform competition, a snapshot mixes proliferation with transport \cite{weinreb2018fundamental}, whereas in a saturated population under uniform competition it approaches the stationary density of the switching dynamics alone after a time of order $1/\lambda$ (Theorems~\ref{thm:disc} and~\ref{thm:cont}). Even then, a static snapshot determines only the effective potential $\Phi$, not $P$ and $D$ separately; separating them requires dynamic information, such as RNA velocity \cite{lamanno2018}, lineage tracing or time-resolved data \cite{weinreb2018fundamental,schiebinger2019}. Methods that learn a potential-driven drift from time series of single-cell snapshots \cite{hashimoto2016,yeo2021}, together with results on when such processes can be recovered from their temporal marginals \cite{lavenant2024}, are suited to this task. We return to this identifiability problem, and to its consequences for cancer, in Section~\ref{sec:examples}.

\medskip
\noindent\textbf{Relation to prior work.} Waddington's metaphor has been formalized in many ways. In the physics literature, cell fates were identified with high-dimensional attractor states of gene-regulatory networks \cite{huang2005}, and cancer with abnormal attractors of the same networks \cite{huang2009}; fate decisions were related to bifurcations of the landscape \cite{ferrell2012} and to transition states between its valleys \cite{moris2016}, and low-dimensional geometric landscapes were fitted to developmental data \cite{corson2012,rand2021,saez2022}; the existence of a global potential for stochastic systems lacking detailed balance was established \cite{ao2007}; landscapes were constructed from steady-state distributions \cite{wang2011}, extended to multi-step differentiation \cite{qiu2012} and to high-dimensional multi-stable systems \cite{tang2017,xu2014}, and compared with the Freidlin--Wentzell quasi-potential of large-deviation theory \cite{zhou2016} (see \cite{wang2015} for a review of the landscape and flux theory). Other works derive effective landscapes from the competition between chromatin-modifying enzymes \cite{alarcon2026}, or from Bayesian decision-making by cells, with regimes corresponding to homeostatic, bistable and cancerous states \cite{entropy2026bayesian}. These approaches work within a postulated SDE/Fokker--Planck formalism, in which potential and noise are posited or inferred phenomenologically; within it, Coomer et al.~\cite{coomer2022} show that very different combinations of drift and noise can produce the same steady-state distribution, the identifiability problem that appears here as the dependence of $\Phi$ on the ratio $P'/D$. Our construction is compatible with this tradition, $P$ acting as an effective quasi-potential, and adds the link to compartmental models: $P$ and $D$ are expressed through the switching rates $k_{ij}$ (Eq.~\eqref{eq:escolhaDv}), via the chain $k_{ij}\to(D,v)\to v=-P'$. A related picture is the ``explore-then-settle'' dynamic of Jiménez-Sánchez et al.~\cite{jimenezsanchez2026}, in which high evolvability drives rapid exploration of the landscape before cells settle at fitness peaks; they treat evolvability as an evolving trait and focus on how exploration and selection interact, whereas we obtain the SGD structure from the gradient reading of the switching velocity and connect it to the population-level Fokker--Planck equation.

\subsection{Non-local transitions and a unified multi-scale model}
\label{sec:nonlocal}

The previous sections assumed that cells move through phenotype space by local exploration resulting from the combination of deterministic drift and random diffusion. The original discrete model \eqref{eq:modelond}, however, allows jumps between any two compartments. Biologically, such non-local transitions are rare but important: cells may make large jumps in extended phenotype space, for instance, through genetic mutations that drastically alter gene expression programs \cite{greavesmaley2012}, through heritable epigenetic changes that silence or activate entire loci \cite{feinberg2016}, or through rare but consequential events such as polyploidization \cite{davoli2011}. We now extend the continuous framework to include these jumps, which leads to a single equation containing all the mechanisms discussed so far (Figure~\ref{fig:roadmap}d,f).

\medskip
\noindent\textbf{From discrete non-local transitions to integro-differential equations.}
When the number of compartments tends to infinity, the discrete transition rates $k_{ij}$ become a continuous kernel $K(x,y) \ge 0$, quantifying the rate at which cells at phenotype $y$ jump to phenotype $x$. The linear transition term in \eqref{eq:modelond} then becomes an integral, and the full continuous model with non-local transitions and uniform competition therefore reads
\begin{equation}\label{eq:integroPDE}
    \partial_t u(x,t) = r(x,t)\, u(x,t) \, g(U(t))
    + \int_\Omega \bigl[ K(x,y) u(y,t) - K(y,x) u(x,t) \bigr] \, dy.
\end{equation}

\noindent\textbf{Kramers--Moyal expansion and Pawula's theorem.} This type of \emph{integro-differential equation} is common in theoretical biology, and is closely related to the advection--diffusion PDE derived earlier.

Suppose that transitions are {localized} in phenotype space, meaning that $K(x,y)$ is sharply peaked around $x \approx y$. More precisely, assume $K(x,y)$ can be written as a function of the jump size $\Delta = x - y$ and the starting point $y$, i.e., $K(x,y) = W(y, \Delta)$ with $W(y, \cdot)$ concentrated near $\Delta = 0$. In this case, we can expand the integral term in a Taylor series in the jump size. This is the \emph{Kramers--Moyal expansion},
\begin{equation}\label{eq:KM}
    \int_\Omega \bigl[ K(x,y) u(y) - K(y,x) u(x) \bigr] \, dy
    = \sum_{n=1}^{\infty} \frac{(-1)^n}{n!}\, \partial_x^n \bigl[ M_n(x)\, u(x) \bigr],
\end{equation}
where $M_n(x) = \int_{-\infty}^{\infty} \Delta^n\, W(x,\Delta)\, d\Delta$ is the $n$-th moment of the jumps that start at $x$. Truncating at second order gives
\[
    -\partial_x\bigl( M_1 u \bigr) + \tfrac12\,\partial_x^2\bigl( M_2 u \bigr)
    = -\partial_x\bigl( v u \bigr) + \partial_x\bigl( D\, \partial_x u \bigr),
\]
with $D(x) = \tfrac12 M_2(x)$ and $v(x) = M_1(x) - D'(x)$, which is exactly the advection--diffusion operator of \eqref{eq:modeloidloc}, or its Fokker--Planck interpretation. The diffusivity is half the second moment of the jump kernel, while the first moment, $M_1 = v + D'$, is the It\^o drift of \eqref{eq:SDEito}; it coincides with the flux-form velocity $v$ only when $D$ is constant.

Truncating the expansion at a higher order does not, however, give a legitimate model. \emph{Pawula's theorem} \cite{pawula1967} states that, for a nonnegative transition probability, the Kramers--Moyal expansion either stops at $n=2$ or contains infinitely many terms: a finite truncation of order $n>2$ cannot be the exact generator of an evolution that preserves positivity, although it may still serve as an approximation \cite{risken1996}. In our modeling context, this rules out any exact intermediate description built from higher-order spatial derivatives (third- or fourth-order terms, say): an exact local, differential description of switching stops at drift and diffusion, and the transitions that it does not capture, the large jumps, must be kept as an integral operator. Pawula's theorem does not imply that a given switching dynamics must be either purely local or purely non-local; local exploration and jumps can coexist, as they do in the model below.

\medskip
\noindent\textbf{Unified macroscopic model.}
We therefore keep local exploration and large jumps as two separate ingredients, each with its own coefficients: a local part, with velocity $v$ and diffusivity $D$ as in \eqref{eq:modeloidloc}, for the small, continual changes of phenotype, and a non-local part, with jump kernel $K$, for the rare, large-effect events. We weight the jump part by a parameter $\varepsilon\ge0$ that sets the rate of jumps relative to local exploration: $\varepsilon=0$ means jumps are absent, and larger $\varepsilon$ means they occur more frequently. This gives a single equation for the population density $u(x,t)$ on a phenotype domain $\Omega\subseteq\mathbb{R}^d$ (we write the multi-dimensional generalization, though our analysis has focused on $d=1$):
\begin{equation}\label{eq:unified}
    \begin{aligned}
        \partial_t u(x,t) &=
        \underbrace{r(x,t) u(x,t) \, g(U(t))}_{\text{reaction (vital dynamics)}}
        - \underbrace{\nabla \cdot \bigl( v(x,t) u(x,t) \bigr)}_{\text{advection (gradient flow)}}
        + \underbrace{\nabla \cdot \bigl( D(x,t) \nabla u(x,t) \bigr)}_{\text{diffusion (random noise)}} \\
        &\quad + \underbrace{\varepsilon \int_\Omega \bigl[ K(x,y) u(y,t) - K(y,x) u(x,t) \bigr] \, dy}_{\text{non-local jumps (mutations, large epigenetic shifts)}} .
    \end{aligned}
\end{equation}
Equation \eqref{eq:unified} is a general macroscopic description of cell populations with phenotypic plasticity. It contains, as special cases:
\begin{itemize}
    \item The discrete compartment model \eqref{eq:modelond} (when $x$ is discretized: the local terms give the nearest-neighbor rates of Section~\ref{sec:continuous}, and $K(x,y)$ the long-range ones);
    \item The reaction--diffusion--advection PDE \eqref{eq:modeloidloc} ($\varepsilon = 0$);
    \item The Fokker--Planck equation of stochastic gradient descent \eqref{eq:FPgrad} ($v = -\nabla P$, $\varepsilon = 0$);
    \item The pure jump integro-differential equation \eqref{eq:integroPDE} ($v\equiv0$, $D\equiv0$, $\varepsilon=1$).
\end{itemize}
Each term encodes a distinct biological process: proliferation and competition (reaction); directed motion toward stable phenotypes (advection); stochastic fluctuations that explore the landscape (diffusion); and abrupt, large-effect events that instantly relocate a cell to a distant region of phenotype space (jumps).

\medskip
\noindent\textbf{The single-cell perspective.}
Through the same duality between Fokker--Planck equations and SDEs discussed in Section \ref{sec:stochastic}, the unified model \eqref{eq:unified} translates into an individual-level description (Figure~\ref{fig:roadmap}f). The phenotype $X_t$ of a single cell follows
\begin{equation}\label{eq:unifiedSDE}
    dX_t = \underbrace{\bigl( -\nabla P(X_t) + \nabla D(X_t) \bigr) \, dt}_{\text{gradient flow}}
    + \underbrace{\sqrt{2D(X_t)} \, dW_t}_{\text{random fluctuations}}
    + \underbrace{dJ^{\varepsilon}_t}_{\text{non-local jumps}},
\end{equation}
where $J^{\varepsilon}_t$ is a pure-jump process. Formally, $J^{\varepsilon}_t$ is constructed from a Poisson random measure with state-dependent intensity $\varepsilon K(y,X_{t^-})\,dy\,dt$ \cite{oksendalsulem2007}.

Equation \eqref{eq:unifiedSDE} generalizes the Ornstein--Uhlenbeck process to arbitrary landscapes with jumps. It combines three ingredients: \emph{order} (the potential $P$), \emph{noise} (the diffusion $D$), and \emph{chance} (the jumps $J^{\varepsilon}_t$). The population-level model \eqref{eq:unified} is recovered as the forward Kolmogorov equation for the density of $X_t$, augmented by a reaction term for proliferation and competition. A cell therefore performs a \emph{stochastic gradient descent with replication and mutation} on the epigenetic landscape.

\medskip
\noindent\textbf{Relation to prior work.} Alvarez et al.~\cite{clairambault2022} and Lorenzi et al.~\cite{villa2025} study integro-differential equations for structured populations, including non-local terms that represent mutations or long-range dispersal. Our unified model \eqref{eq:unified} is closely related to these formulations; here the local part is written in terms of a potential and the non-local part is weighted by a parameter $\varepsilon$. The connection between jump processes and diffusion equations via the Kramers--Moyal expansion is a classical result in stochastic processes \cite{vankampen1992}. That drift, diffusion and jumps exhaust the possibilities follows from Courr\`ege's theorem: under mild regularity conditions, the generator of a Markov process whose evolution preserves positivity (in the sense of the positive maximum principle) consists of a drift, a diffusion and a jump part \cite{courrege1965,bottcher2013}, a state-dependent extension of the L\'evy--Khintchine formula \cite{applebaum2009}, and the unified model has exactly this form. The irrelevance of fitness under uniform competition also holds for purely non-local switching ($v\equiv0$, $D\equiv0$ in \eqref{eq:unified}), provided the jump kernel connects every phenotype to a common region within a finite number of jumps (the analogue of (H3)), and for the unified model \eqref{eq:unified} itself, in which jumps are combined with drift and diffusion; the limit is then the stationary density of the full switching dynamics \cite{companion2026}. Finally, jumps need not be rare. Proteins are often produced in bursts, and the classical model of Friedman, Cai and Xie \cite{friedman2006}, with linear degradation and exponentially distributed bursts, is an instance of \eqref{eq:unified} with a linear drift $v(x)=-\gamma x$, no diffusion ($D\equiv0$) and a one-sided jump kernel, $K(x,y)=k\,b^{-1}e^{-(x-y)/b}$ for $x>y$ (bursts at rate $k$ with mean size $b$); its stationary density is a gamma distribution, a shape commonly observed for protein abundances \cite{taniguchi2010}. Bursts are thus frequent jumps in expression space, whereas mutations are rare jumps in an extended phenotype space.

\subsection{Examples: phenotypic distributions emerging from the landscape}
\label{sec:examples}

We now show with simulations (Figure~\ref{fig:examples}) how deterministic drift and random motion together shape the phenotypic distribution, as sketched for single cells in Figure~\ref{fig:roadmap}e,f, and produce several well-known probability laws.

Throughout this section we assume that the population has reached the equilibrium $U^*$, so that the distribution is governed solely by the transport terms. Recall that the stationary density $u_{\text{eq}}(x)$ satisfies \eqref{eq:stationary}, which, using $v(x) = -P'(x)$, becomes
\begin{equation}\label{eq:stationary2}
    u_{\text{eq}}(x) = C \exp\!\Bigl( -\int_{x_0}^x \frac{P'(s)}{D(s)}\,ds \Bigr),
\end{equation}
with $C$ a normalization constant.

\medskip
\noindent\textbf{Case 1: Flat landscape -- uniform distribution.}
When the epigenetic landscape is flat, $P(x) \equiv \text{const}$, the velocity vanishes ($v \equiv 0$). Cells move purely by random fluctuations, and \eqref{eq:stationary2} yields $u_{\text{eq}}(x) = C$, a \emph{uniform distribution} on the domain $\Omega$. All phenotypes are equally probable; microscopically, the SDE reduces to a pure Brownian motion $dX_t = \sqrt{2D}\,dW_t$, and the cell wanders without directional bias (Figure~\ref{fig:examples}a).

\begin{figure}[htbp]
\centering
\includegraphics[width=0.95\textwidth]{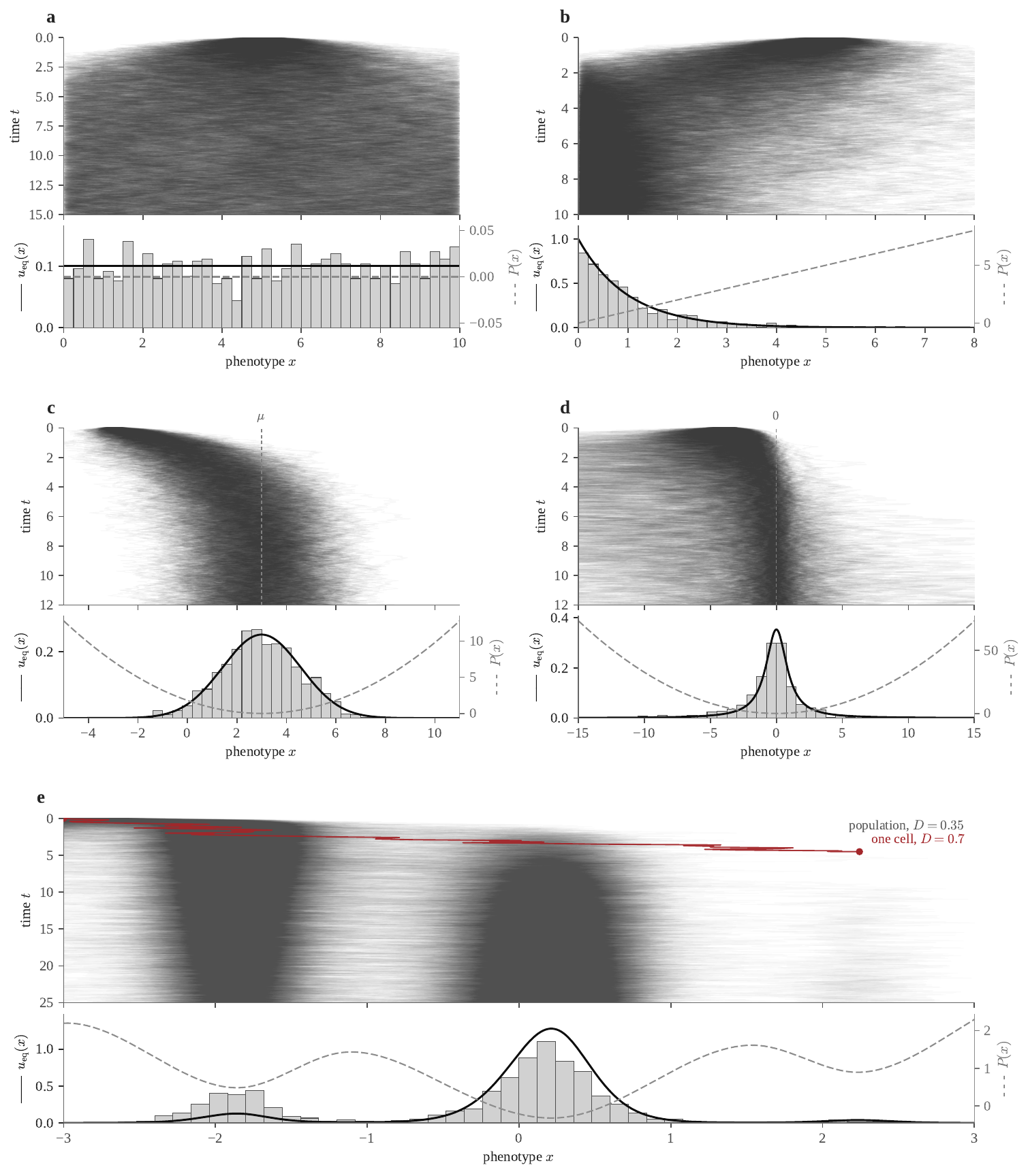}
\caption{\footnotesize\linespread{1}\selectfont\textbf{Known distributions emerging from SGD on different potential landscapes.} Direct simulation of the single-cell SDE \eqref{eq:SDEito}. In each panel, the top plot shows $N$ individual trajectories (Euler--Maruyama, phenotype on the horizontal axis, time on the vertical axis plotted downward so that its bottom edge aligns with the final-time histogram below); the bottom plot shows the histogram of final positions, the analytic stationary density (solid black, left axis), and the potential $P(x)$ (dashed gray, right axis). \textbf{(a)} Case 1: flat landscape, uniform distribution on $[0,10]$ with reflecting boundaries, $D=1$, $N=1000$. \textbf{(b)} Case 2: linear potential $P(x)=ax$, exponential distribution~\eqref{eq:exponential} with $a=D=1$ and a reflecting boundary at $x=0$, $N=1000$. \textbf{(c)} Case 3: harmonic potential~\eqref{eq:harmonic_potential}, Gaussian distribution~\eqref{eq:gaussian} with $\theta=0.4$, $\mu=3$, $D=1$, $N=1000$. \textbf{(d)} Case 4: state-dependent noise $D(x)=D_0(1+x^2)$ with $D_0=0.3$, $\theta=0.65$, giving the heavy-tailed distribution~\eqref{eq:heavytail}, plotted on a wider horizontal scale than (c), $N=1000$. \textbf{(e)} The general Boltzmann-form distribution~\eqref{eq:boltzmann} for the rugged, multi-well landscape sketched schematically in Figure~\ref{fig:roadmap}e,f ($D=0.35$, constant); all $N=1500$ trajectories start at the left edge ($x=-3$), explore the landscape through noise-driven barrier crossings, and settle partially in the first well, predominantly in the deepest one, and rarely in the shallow well on the far right; at $t=25$ the population is still in transit, with the left well overpopulated relative to $u_{\text{eq}}$. Compare with the continuous counterpart, Figure~\ref{fig:pdesim}c. The colored line is one cell with twice the noise ($D=0.7$), a more plastic cell, drawn until it reaches the far well at $t\approx4.5$; it is a fast realization chosen for illustration (typical first-passage times are given in the text). Euler--Maruyama time step $\Delta t=0.01$ ($0.02$ in (c), $0.004$ in (d)); reflecting boundaries are implemented by mirror reflection, and in (e) the drift is extended as a constant outside $[-3,3]$.}
\label{fig:examples}
\end{figure}

\medskip
\noindent\textbf{Case 2: Linear potential -- exponential distribution.}
Consider a linear potential $P(x) = a x$ with $a > 0$, giving a constant negative velocity $v(x) = -a$ that drives cells toward smaller $x$. With constant diffusivity $D$ on $\Omega = [0,\infty)$ and a reflecting boundary at $x = 0$, the stationary solution becomes
\begin{equation}\label{eq:exponential}
    u_{\text{eq}}(x) = \frac{a}{D}\, e^{-(a/D)x}.
\end{equation}
This is an \emph{exponential distribution} with rate $a/D$, i.e., mean $D/a$. The drift concentrates the population near the origin, while diffusion allows occasional escapes to larger $x$. Biologically, this describes cells pushed toward a boundary phenotype (e.g., a fully differentiated state) by a potential of constant slope, with noise providing the observed tail of less-differentiated cells. At the single-cell level, $dX_t = -a\,dt + \sqrt{2D}\,dW_t$: the cell experiences a constant downward drift perturbed by random fluctuations (Figure~\ref{fig:examples}b).

\medskip
\noindent\textbf{Case 3: Harmonic potential -- Gaussian distribution.}
An important case is the harmonic potential
\begin{equation}\label{eq:harmonic_potential}
    P(x) = \frac{\theta}{2}(x - \mu)^2,
\end{equation}
for which $v(x) = -P'(x) = -\theta(x - \mu)$, with $\theta > 0$ and a reference phenotype $\mu$. This is the velocity field of the Ornstein--Uhlenbeck process \eqref{eq:OUSDE}. With constant diffusivity $D$, equation \eqref{eq:stationary2} gives the \emph{Gaussian} stationary distribution
\begin{equation}\label{eq:gaussian}
    u_{\text{eq}}(x) = \frac{1}{\sqrt{2\pi \sigma^2}} \exp\!\Bigl( -\frac{(x - \mu)^2}{2\sigma^2} \Bigr), \qquad \sigma^2 = \frac{D}{\theta}.
\end{equation}
The mean phenotype $\mu$ is the minimum of the potential, i.e., the ``center of attraction'', $\theta$ controls the steepness of the valley, and $D$ is the noise intensity. The variance $\sigma^2 = D/\theta$ expresses the balance between noise and order: a steep valley (large $\theta$) produces a tightly concentrated population; strong noise (large $D$) produces a broad distribution (Figure~\ref{fig:examples}c).

\medskip
\noindent\textbf{Universality of the Gaussian.}
The Gaussian is, in a precise sense, the \emph{generic} local approximation near any stable equilibrium. If $P(x)$ has a non-degenerate minimum at $x = \mu$, then $P'(\mu) = 0$ and $P''(\mu) > 0$. Expanding to second order yields $P(x) \approx P(\mu) + \frac{1}{2}P''(\mu)(x-\mu)^2$, which is precisely the harmonic potential with $\theta = P''(\mu)$ capturing the curvature of the landscape around $x=\mu$, with the Ornstein--Uhlenbeck velocity field $v(x)=-P'(x)\approx-P''(\mu)(x-\mu)$. To leading order the diffusivity is constant, $D(x) \approx D(\mu)$, and the stationary distribution is Gaussian with variance $D(\mu)/P''(\mu)$. Higher-order terms in the expansion of $P(x)$ introduce polynomial corrections to the exponent, deforming the tails. In this sense, the Ornstein--Uhlenbeck process of Section~\ref{sec:stochastic} is the local linearization of the general stochastic gradient descent \eqref{eq:SDEgradient} around any stable phenotype: every minimum of the epigenetic landscape looks, up close, like a quadratic well with OU dynamics. A rugged epigenetic landscape with multiple minima can thus be approximated locally by a mixture of Gaussians, each centered at a stable phenotypic state with variance inversely proportional to the local curvature of the potential (Figure~\ref{fig:examples}e). Measured distributions of protein abundance are nevertheless often skewed, closer to gamma or log-normal than to Gaussian \cite{taniguchi2010}: the harmonic approximation describes weak fluctuations around a stable phenotype, not every source of variability, and bursty production, a form of frequent jumps, produces gamma distributions (Section~\ref{sec:nonlocal}).

\medskip
\noindent\textbf{Case 4: State-dependent noise -- heavy tails and effective potentials.}
When the diffusivity $D(x)$ is not constant, the stationary distribution can be far from Gaussian. A biologically motivated example is $D(x) = D_0(1 + x^2)$ with $x\in (-\infty,+\infty)$, representing increasing phenotypic instability away from the origin (e.g., higher epigenetic noise in less-differentiated states). With a harmonic potential $P(x) = \frac{\theta}{2}x^2$, the stationary solution, which would be a Gaussian centered at $x=0$ in case of constant diffusivity, becomes
\begin{equation}\label{eq:heavytail}
    u_{\text{eq}}(x) = C\,(1 + x^2)^{-\theta/(2D_0)}.
\end{equation}
This is a \emph{Student's t-like} distribution with power-law tails that decay far more slowly than a Gaussian (and is normalizable provided $\theta>D_0$). The population maintains a non-negligible fraction of cells far from the mean phenotype even under a restoring force. Where rare, distant subpopulations of this kind are observed, this example shows that state-dependent noise, and not only a deformed potential, can produce them. Compared with Case 3, the same restoring force gives a sharper central peak and much longer tails, only because the noise intensity grows with distance from the optimum (Figure~\ref{fig:examples}d, plotted on a wider horizontal scale than panel c to accommodate the tails).

\medskip
\noindent\textbf{The effective potential: a general Boltzmann form.}
This example is a special case of a general formula. For arbitrary $D(x)$, the stationary distribution can be written in the \emph{Boltzmann form}
\begin{equation}\label{eq:boltzmann}
    u_{\text{eq}}(x) = C \exp\!\bigl( -\Phi(x) \bigr),
\end{equation}
where the \emph{effective potential} $\Phi(x)$ satisfies $\Phi'(x) = P'(x)/D(x)$. For a rugged, multi-well landscape and constant $D$, this is simply the familiar $\exp(-P(x)/D)$, and simulating it directly (Figure~\ref{fig:examples}e) shows a population initialized far from equilibrium exploring the landscape by repeated noise-driven barrier crossings and settling, in the long run, into each well in proportion to its Boltzmann weight $\int_{\text{well}}e^{-P/D}\,dx$. At the final time shown, $t=25$, the population is still in transit: the left well is overpopulated relative to $u_{\text{eq}}$, and the far, shallow well is reached only on the time scale of the first-passage times discussed below. This is the single-cell counterpart of Figure~\ref{fig:pdesim}, where the population density converged to the same stationary distribution on the same landscape.

\medskip
\noindent\textbf{Noise versus landscape.}
When $D(x)$ varies, the landscape $P(x)$ is reshaped by the local intensity of phenotypic noise into the effective potential $\Phi(x)$. Where $D(x)$ is elevated, reflecting for instance a relaxation of the biochemical constraints that normally restrict fluctuations, $\Phi$ is flattened relative to $P$ and the barriers between phenotypic states are effectively lowered. This is an alternative, and perhaps underappreciated, route to increased plasticity: rather than altering the landscape $P(x)$ itself, a cell population may simply become ``noisier'' in specific regions of phenotype space. The effect on single cells can be large: on the multi-well landscape of Figure~\ref{fig:examples}e, doubling the noise shortens the typical time needed to reach the distant, shallow well roughly tenfold, from $\approx250$ to $\approx22$ time units (median first-passage times over $10^3$ cells), because barrier crossing is exponentially sensitive to the ratio of barrier height to noise (the Kramers effect). The converse is an identifiability problem. Two populations with the same potential $P(x)$, that is, the same stable states and barriers, but different noise profiles $D(x)$ have different stationary distributions, and a population with elevated noise in some region looks, in a static snapshot, as if its barriers there had been lowered; inferring $P$ from $u_{\text{eq}}$ therefore requires knowing $D$ (Section~\ref{sec:gradient}). This is particularly relevant in cancer, where the fidelity of epigenetic maintenance is compromised and increased stochastic variability of DNA methylation has been documented across tumor types \cite{feinbergirizarry2010,hansen2011,landau2014}: what appears as a flattened landscape in single-cell data may reflect increased noise rather than a reshaping of $P(x)$.

\medskip
\noindent\textbf{Mobility, temperature, and the choice of convention.}
A physical parametrization makes precise what a snapshot can and cannot reveal. For a Brownian particle, the diffusivity is the product of a mobility $\mu$, which measures how fast the particle responds to a force, and a temperature $T$, which measures the intensity of the thermal fluctuations: $D=\mu T$ (the Einstein relation \cite{einstein1905}). By analogy, write $D=\mu T$ and $v=-\mu E'$, with a mobility $\mu(x)$, an effective temperature $T(x)$ and an energy landscape $E(x)$, so that $P'=\mu E'$ \cite{lau2007}. In the flux form, $u_{\text{eq}}\propto\exp\bigl(-\int E'/T\bigr)$: the mobility cancels, and the effective potential is $\Phi'=E'/T$. Case~4 is then a temperature that varies across phenotype space at fixed mobility, which is legitimate for a system out of equilibrium. Elevated plasticity accordingly has two sub-routes. A population can become \emph{hotter}, which flattens $u_{\text{eq}}$ and, in a snapshot, is confounded with a change of the landscape; or it can become \emph{more mobile}, with $T$ and $E$ unchanged, which leaves $u_{\text{eq}}$ unchanged and is invisible in any snapshot, but speeds up all transitions. Since the mixing time $1/\lambda$ cannot increase when $\mu$ increases, a more mobile tumor also forgets past selection faster (Section~\ref{sec:uniform_competition}).

\medskip
\noindent\textbf{Relation to prior work.} The Boltzmann-form stationary distribution \eqref{eq:boltzmann} underlying every case in this section is the same object constructed, from the opposite direction, by the landscape-flux literature discussed in Section~\ref{sec:gradient}: Wang et al.~\cite{wang2011} build the potential directly from the steady-state distribution of a postulated gene-regulatory SDE, whereas here $u_{\text{eq}}(x)$ is a consequence of a potential and diffusivity that are themselves derived from microscopic switching rates. Case 3 and the argument on the universality of the Gaussian recover, as a special case, the classical result that stabilizing selection around a fixed optimum yields a Gaussian phenotypic distribution, as in the Ornstein--Uhlenbeck models of de Souza Silva et al.~\cite{desouzasilva2023} discussed in Section~\ref{sec:stochastic}. Case 4 can be compared with other routes to non-Gaussian phenotypic distributions in the literature reviewed here, which rest on different mechanisms: de Souza Silva et al.~\cite{desouzasilva2023} obtain skewed distributions from a {moving} fitness optimum combined with phenotypic plasticity; Kessler and Levine~\cite{kessler2022} obtain a non-Gaussian (Airy-function) steady state from a {nonlinear}, replicator-mutator-type reaction term combined with a reflecting boundary. Here, the potential $P(x)$ and the restoring dynamics remain exactly as in the Gaussian case, and non-Gaussianity may be generated purely by a state-dependent diffusivity $D(x)$. These three mechanisms are not mutually exclusive, and a full account of a given biological system may require more than one of them; distinguishing between them requires dynamic information.

\subsection{Summary}

This section went from simple two-compartment models of quiescent vs. proliferating cells to a unified framework. Under uniform competition, and on time scales long compared with the mixing time of switching, long-term heterogeneity is governed by transitions and not by fitness. Local transitions correspond to a finite-difference discretization of a reaction--diffusion--advection PDE, with diffusion and advection carried by the symmetric and skew-symmetric parts of the transition matrix. The advection velocity can be interpreted as the negative gradient of an effective epigenetic potential. The resulting model, translated to the individual-cell trajectory, is a stochastic gradient descent on an effective potential landscape. Including occasional jumps, the full model \eqref{eq:unified}-\eqref{eq:unifiedSDE} connects the discrete and continuous descriptions, and the single-cell and population levels.

Taken individually, several components of our framework have antecedents in the literature. What is new here is their assembly into a single chain of arguments, from discrete switching rates to a continuous landscape and to stochastic gradient descent, summarized in Figure~\ref{fig:roadmap}. Some steps of this chain rest on assumptions that are stated explicitly above, notably uniform competition, the gradient form of the velocity, and the choice of stochastic convention for state-dependent noise.

\section{Conclusion and outlook}
\label{sec:conclusion}

We first qualify the main statement about fitness and then summarize the arguments above in a few principles.

\medskip
\noindent\textbf{How general is the irrelevance of fitness?} Theorems~\ref{thm:disc} and~\ref{thm:cont} make two statements: uniform, multiplicative competition makes all phenotypes selectively neutral at saturation, and the selection that acts during growth leaves a transient imprint which switching erases. Fitness returns whenever one of these ingredients fails: under other forms of density dependence, such as additive competition, and when switching is coupled to division (Sections~\ref{sec:uniform_competition} and~\ref{sec:pdeasymp}). Under non-uniform competition, convergence itself can fail (Remark~\ref{rem:oscillations} in Appendix~A). The statement is also asymptotic: it holds on time scales long compared with the mixing time $1/\lambda$ of the switching dynamics, which can be very long in landscapes with high barriers, and it is on this time scale that relapse unfolds after a treatment has ended. We therefore regard the ``primacy of plasticity'' as a property of uniformly competitive models in their long-time limit, not as a general biological law; which regime applies to a given tumor is an empirical question.

\medskip
\noindent\textbf{Connection with machine learning.} In principle, the expression ``stochastic gradient descent'' here plays a different role than in machine learning. In machine learning, the randomness of stochastic gradient descent comes from estimating the gradient on random subsets of the data, whereas here, in equation \eqref{eq:SDEgradient}, it is an explicit noise that lets cells cross the barriers between minima. The two are nevertheless related: a line of work has shown that, under suitable assumptions, gradient-based learning algorithms with noise, whether added explicitly or arising from subsampling, behave approximately as discretizations of Langevin equations of the type \eqref{eq:SDEgradient} \cite{wellingteh2011,mandt2017,smithle2018,malladi2022}, although the correspondence is only approximate \cite{chaudhari2018}.

The same equation nevertheless offers a common language for stochastic optimization and for phenotype dynamics on a landscape: in both, noise helps escape local minima, the landscape can be reshaped by external signals, and occasional large jumps give access to distant regions of the space. These parallels are heuristic. Mathematics and biology have long borrowed from each other \cite{cohen2004mathematics}, and artificial neural networks were themselves inspired by biological neurons; the similarity with the epigenetic dynamics of a cell should be read as an analogy, not as evidence of a shared mechanism.

\medskip
\noindent\textbf{Three rules for complexity: life as a stochastic gradient descent with occasional jumps.} In this light, the single-cell trajectory \eqref{eq:unifiedSDE} can be read as a stochastic gradient descent, in the Langevin sense, with noise that lets the cell escape local minima and occasional jumps that relocate it across the landscape. Evolution, development, and disease are then stochastic search processes on a landscape, with replication amplifying the cells that find favorable regions. Three elementary ingredients govern these dynamics:

\begin{enumerate}
    \item {Minimize energy.} Biological entities, from folding proteins to gene-regulatory networks, relax toward stable configurations: minima of a free energy for passive systems, and of an effective potential for driven, nonequilibrium ones. In the phenotypic realm, this is the drift $v = -\nabla P$ toward stable expression programs; alone, it would leave each cell trapped in the nearest valley, with no diversity and no adaptation.

    \item {Allow randomness.} Stochastic fluctuations (thermal noise, transcriptional bursting \cite{elowitz2002}, unequal partitioning at division) provide the diffusion that lets cells explore the landscape; they contribute to the diversification of cell types during development, alongside deterministic mechanisms such as bifurcations, morphogen gradients and asymmetric divisions, and maintain the heterogeneity that keeps populations resilient. Without the drift, noise alone would spread cells over the whole phenotype space, with no stable cell types.

    \item {Permit rare jumps.} Mutations, large-scale epigenetic reprogramming, and other non-local events give access to regions of the landscape that local exploration would reach only after astronomically long times, if at all. They are a source of evolutionary novelty, and also of malignant transformation.
\end{enumerate}

This triad, \emph{gradient flow, diffusion, and jumps}, corresponds to the three transport terms of the unified model \eqref{eq:unified}, and it is not an ad hoc collection of mechanisms: by Courr\`ege's theorem, under mild regularity conditions, the generator of a Markovian switching dynamics consists of a drift, a diffusion and a jump part (Section~\ref{sec:nonlocal}) \cite{courrege1965,bottcher2013}, and, by Pawula's theorem, the jump part cannot be replaced exactly by higher-order local terms.

\medskip
\noindent\textbf{Cancer as the corruption of the landscape.}
In this picture, cancer can disturb each of the three rules, and the framework separates the corresponding routes. Mutations in epigenetic regulators (writers, erasers, remodelers) reshape the landscape $P(x)$, creating pathological minima that correspond to proliferative or resistant states, in line with the view of cancer states and drug resistance as attractors of the underlying networks \cite{huang2009,pisco2013,pisco2015}. Elevated plasticity can instead correspond to a larger diffusivity $D(x)$ with the landscape intact, consistent with the increased epigenetic variability of tumors \cite{hansen2011,landau2014}: either a higher effective temperature, which lowers the barriers effectively and flattens the stationary density, or a higher mobility, which leaves the stationary density unchanged and speeds up transitions (Section~\ref{sec:examples}); this route requires no mutation in chromatin modifiers and may be invisible to genomic sequencing, yet detectable, with dynamic data, in single-cell distributions (Section~\ref{sec:gradient}). If the extra variability arises from errors made at replication, however, switching is coupled to division, and fitness regains its influence on the long-term composition (Section~\ref{sec:uniform_competition}). Driver mutations, finally, act as non-local jumps that seed cells in distant regions of the landscape (Section~\ref{sec:nonlocal}). Which of these routes dominates in a given tumor, and on which time scale selection or plasticity shapes the resulting heterogeneity (Sections~\ref{sec:uniform_competition} and~\ref{sec:pdeasymp}), are empirical questions that the framework helps to formulate.

\medskip
\noindent\textbf{Future directions.}
Several directions remain open:
\begin{itemize}
    \item {Inference from data:} Single-cell multi-omics data (scRNA-seq, scATAC-seq, single-cell methylation) can be used to reconstruct $P(x)$, $D(x)$ and $K(x,y)$, and the diversity indices used to quantify intratumor heterogeneity \cite{ferrallfairbanks2019} become functionals of the stationary density $\psi$: a noisier population has a flatter $\psi$ and a higher Shannon diversity. Separating $P$ from $D$ requires dynamic data (Section~\ref{sec:examples}), and identifiability analyses indicate which experimental designs suffice \cite{browning2019,browning2025identifiability}.
    \item {Control:} If the landscape can be inferred, one can ask how to manipulate it. For example, by designing epigenetic therapies that reshape $P(x)$ to eliminate resistant minima, or by timing treatments to exploit the stochastic dynamics of persistence.
    \item {Multi-dimensional extensions:} Real phenotypic landscapes are high-dimensional. Extending the analysis to $d>1$, where non-gradient drifts arise (Section~\ref{sec:gradient}), and connecting it to manifold-learning techniques such as diffusion maps is a natural next step; moment-closure reductions of phenotype-structured PDEs to low-dimensional ODE systems \cite{villa2025moments} may keep such extensions computationally tractable.
    \item {Beyond the present theorems:} Theorems~\ref{thm:disc} and \ref{thm:cont} assume time-independent switching and smooth coefficients; in the continuum, the result covers one-dimensional phenotypes and bounded phenotype domains in higher dimension \cite{companion2026}. Extending them to rough coefficients, to unbounded higher-dimensional phenotype spaces, and to switching rates that change in time (for instance under treatment) is open. So is the question, raised by the counterexample of Appendix~A, of which forms of non-uniform competition still lead to convergence, and how close to $\pi$ the composition remains when competition is nearly uniform.
    \item {Spatial aspects and non-uniform competition:} In solid tumors, spatial heterogeneity and local competition for resources may violate the uniform competition hypothesis. Extending the framework to physical space and non-uniform competition would show how local selection and long-range transport interact; in this setting, fitness differences may re-enter the asymptotic dynamics alongside the transport terms.
    \item {Non-Markovian dynamics and memory effects:} Transition rates may depend on a cell's history, as in the ``hypoxic memory'' of cells that retain an invasive, slow-cycling phenotype after reoxygenation, recently incorporated into a phenotype-structured PDE \cite{sadhu2026}. Allowing $P$ or $D$ to depend on past states would connect the framework to generalized Langevin dynamics, at the cost of the direct Fokker--Planck correspondence used here.
\end{itemize}

Under the assumptions stated along the way, Waddington's landscape can thus be given a concrete mathematical form, in which each cell performs a stochastic gradient descent, with occasional jumps, on a landscape of phenotypic states.

\section*{Appendix}
\addcontentsline{toc}{section}{Appendix}

\subsection*{A. Proof of Theorem~\ref{thm:disc}}

The proof uses two elementary tools: the variation-of-constants formula for linear ODEs, and the fact that an irreducible Markov chain forgets its initial condition. The latter is classical (see, e.g., \cite{norris1997}); we include a short proof to keep the argument self-contained. Throughout, $\abs{v}_1=\sum_i\abs{v_i}$, $\one=(1,\dots,1)^\top$, and a vector is \emph{zero-sum} if $\sum_iv_i=0$.

\paragraph{A preliminary fact.} Every column of $A$ sums to zero: $\sum_iA_{ij}=0$ for each $j$, because every cell that leaves compartment $j$ enters some other compartment, so what compartment $j$ loses is exactly gained by the others. In vector notation, $\one^\top A=0$. Differentiating $\one^\top P(t)=\one^\top e^{tA}$ in $t$,
\[
\frac{d}{dt}\bigl(\one^\top P(t)\bigr)=\one^\top A\,P(t)=0,
\]
and at $t=0$, $\one^\top P(0)=\one^\top I=\one^\top$. Hence $\one^\top P(t)=\one^\top$ for every $t\ge0$: every column of $P(t)=e^{tA}$ sums to $1$ and, together with positivity (proved below), is a probability vector.

\begin{lemma}[Switching forgets the initial state]\label{lem:mix}
Assume \textup{(H3)} and let $P(t)=e^{tA}$ be the solution operator of $u'=Au$.
\begin{enumerate}[label=(\roman*), itemsep=2pt]
\item $P(t)$ maps nonnegative vectors to nonnegative vectors and preserves the total, $\sum_i(P(t)v)_i=\sum_iv_i$. In particular $\abs{P(t)v}_1\le\abs{v}_1$ for every $v$.
\item For every $\tau>0$ all entries of $P(\tau)$ are positive.
\item There are $C\ge1$, $\lambda>0$ such that $\abs{P(t)w}_1\le Ce^{-\lambda t}\abs{w}_1$ for all zero-sum $w$ and $t\ge0$.
\item There is a unique vector $\pi$ with $A\pi=0$ and $\sum_i\pi_i=1$, and all its entries are positive.
\end{enumerate}
\end{lemma}

\begin{proof}
\emph{(i) Positivity and mass conservation.} Let $c=\max_i\sum_{l\ne i}k_{il}$, the largest total exit rate over all compartments, and $B=A+cI$. Every off-diagonal entry of $B$ equals that of $A$, hence is $\ge0$, and every diagonal entry is $A_{ii}+c=c-\sum_{l\ne i}k_{il}\ge0$ by the choice of $c$; so $B\ge0$ entrywise. Then
\[
P(t)=e^{tA}=e^{-ct}e^{tB}=e^{-ct}\sum_{k\ge0}\frac{t^kB^k}{k!}
\]
is a (convergent) sum of nonnegative matrices scaled by $e^{-ct}>0$, hence $P(t)\ge0$ entrywise. Mass conservation is the preliminary fact above, $\one^\top P(t)=\one^\top$, i.e.\ $\sum_i(P(t)v)_i=\sum_iv_i$ for every vector $v$. For the contraction bound, write $(P(t)v)_i=\sum_jP(t)_{ij}v_j$; by the triangle inequality and $P(t)_{ij}\ge0$,
\[
\abs{P(t)v}_1=\sum_i\Bigl|\sum_jP(t)_{ij}v_j\Bigr|\ \le\ \sum_i\sum_jP(t)_{ij}\abs{v_j}\ =\ \sum_j\abs{v_j}\underbrace{\sum_iP(t)_{ij}}_{=\,1\text{ by mass conservation}}\ =\ \abs{v}_1 .
\]

\emph{(ii) Strict positivity of $P(\tau)$.} With $B$ as above, the entry $(B^k)_{ij}$ is a sum of products $B_{i\,l_1}B_{l_1l_2}\cdots B_{l_{k-1}j}$ over paths of length $k$ from $j$ to $i$, and it is positive as soon as one such path uses only positive rates $k_{lm}>0$ (recall $B_{lm}=A_{lm}=k_{ml}$ for $l\ne m$). By (H3) every compartment $i$ can be reached from every compartment $j$ through such a path, of some length $k=k(i,j)$; and $(B^0)_{ii}=1>0$ covers $i=j$. Hence every entry of $e^{\tau B}=\sum_k\tau^kB^k/k!$ is positive (each term is $\ge0$ and at least one term is strictly positive at each entry), and so is every entry of $P(\tau)=e^{-c\tau}e^{\tau B}$.

\emph{(iii) Exponential mixing on zero-sum vectors.} Let $\delta:=\min_{i,j}P(1)_{ij}>0$, positive by (ii). Since each column of $P(1)$ sums to $1$ and has all $n$ entries $\ge\delta$, we have $n\delta\le1$; if $n\delta=1$ replace $\delta$ by $\delta/2$, so that from now on $n\delta<1$ (this can only weaken the bound obtained below). Let $w$ be zero-sum and split it into its positive and negative parts, $w=w^+-w^-$, with $w^\pm\ge0$. Since $\sum_iw_i=0$, the two parts carry the same total mass:
\[
\sum_iw_i^+=\sum_iw_i^-=\tfrac12\abs{w}_1 .
\]
Because every entry of $P(1)$ is at least $\delta$ and $w_j^+\ge0$,
\[
\bigl(P(1)w^+\bigr)_i=\sum_jP(1)_{ij}w_j^+\ \ge\ \delta\sum_jw_j^+=\tfrac\delta2\abs{w}_1\qquad\text{for \emph{every} compartment }i,
\]
and likewise for $w^-$. So both vectors
\[
A^\sharp:=P(1)w^+-\tfrac\delta2\abs{w}_1\,\one,\qquad B^\sharp:=P(1)w^--\tfrac\delta2\abs{w}_1\,\one
\]
have all entries $\ge0$. By mass conservation applied to $w^+$ and $w^-$ separately, their totals are
\[
\sum_iA^\sharp_i=\sum_iB^\sharp_i=\tfrac12\abs{w}_1-n\cdot\tfrac\delta2\abs{w}_1=\tfrac12\abs{w}_1(1-n\delta),
\]
and, since $A^\sharp,B^\sharp\ge0$, this total is their $\ell^1$ norm. Since $P(1)w=A^\sharp-B^\sharp$, the triangle inequality gives
\[
\abs{P(1)w}_1\le\abs{A^\sharp}_1+\abs{B^\sharp}_1=(1-n\delta)\abs{w}_1 .
\]
Moreover $P(1)w$ is again zero-sum, by mass conservation, so the bound can be iterated:
\[
\abs{P(1)^mw}_1\le(1-n\delta)^m\abs{w}_1\qquad\text{for every integer }m\ge0 .
\]
For general $t\ge0$ write $t=m+s$ with $m=\lfloor t\rfloor$ and $0\le s<1$. By the semigroup property $P(t)=P(s)P(1)^m$ and part (i) (which needs no zero-sum condition),
\[
\abs{P(t)w}_1\ \le\ \abs{P(1)^mw}_1\ \le\ (1-n\delta)^m\abs{w}_1\ \le\ (1-n\delta)^{t-1}\abs{w}_1 ,
\]
using $m\ge t-1$ and $0<1-n\delta<1$. This is (iii), with $C=(1-n\delta)^{-1}$ and $\lambda=-\ln(1-n\delta)>0$.

\emph{(iv) Existence, uniqueness and positivity of $\pi$.} Let $q$ be any probability vector. For $t,s\ge0$, the semigroup property gives $P(t+s)q-P(t)q=P(t)\bigl(P(s)q-q\bigr)$. The vector $P(s)q-q$ is zero-sum (both terms are probability vectors) and $\abs{P(s)q-q}_1\le2$ by part (i). Applying (iii),
\[
\abs{P(t+s)q-P(t)q}_1\ \le\ 2Ce^{-\lambda t}\qquad\text{for \emph{every} }s\ge0 .
\]
This is the Cauchy criterion for the curve $t\mapsto P(t)q$ in $\R^n$, so $P(t)q$ converges to some limit $\pi$ as $t\to\infty$; being a limit of probability vectors, $\pi\ge0$ and $\sum_i\pi_i=1$. Taking $t\to\infty$ in $P(t+s)q=P(s)P(t)q$ gives $\pi=P(s)\pi$ for every $s\ge0$, and differentiating at $s=0$ gives $A\pi=0$. Positivity follows from $\pi=P(1)\pi$ and (ii): $\pi_i=\sum_jP(1)_{ij}\pi_j\ge\delta>0$. Finally, if $\pi'$ is another stationary probability vector, then $\pi'-\pi$ is zero-sum and fixed by every $P(t)$, so by (iii) $\abs{\pi'-\pi}_1\le Ce^{-\lambda t}\abs{\pi'-\pi}_1\to0$, i.e.\ $\pi'=\pi$.
\end{proof}

\begin{proof}[Proof of Theorem~\ref{thm:disc}]
\emph{Step 1: the total population is monotone.} Solutions of \eqref{eq:modelondMat} stay nonnegative, because a compartment that reaches zero can only receive cells: if $u_i=0$ and all other $u_j\ge0$, then $u_i'=\sum_{j\ne i}k_{ji}u_j\ge0$ (see, e.g., \cite{smith1995}). Summing the equations, the transition terms cancel and
\begin{equation}\label{eq:U}
U'=\rho(t)\,g(U),\qquad\rho(t)=\sum_ir_i(t)\,u_i(t)\ge0 .
\end{equation}
Regard $\rho$ as a given function of time; it is only piecewise continuous when the $r_i$ are, and \eqref{eq:U} is then understood in the sense of Carath\'eodory. The constant $U^*$ is a solution of \eqref{eq:U} (since $g(U^*)=0$), and, since $g$ is locally Lipschitz, solutions of \eqref{eq:U} are unique, so that two of them cannot cross. Hence $U(t)-U^*$ never changes sign, so neither does $g(U(t))$, nor $U'$. Thus $U$ is monotone and lies between $U(0)$ and $U^*$; in particular $U(t)\ge m:=\min\{U(0),U^*\}>0$, the solution exists for all times, $U(t)$ converges to a limit $U_\infty$, and
\begin{equation}\label{eq:TV}
\int_0^\infty\abs{U'(s)}\,ds=\abs{U_\infty-U(0)}<\infty .
\end{equation}

\emph{Step 2: the reaction term has one sign.} Let $F=g(U)R(t)u$ be the reaction term, so that $u'=Au+F$. Together with Step~1, where it gives the scalar equation \eqref{eq:U}, this is where (H1) is used: since $r_i\ge0$, $u_i\ge0$, and all compartments share the same factor $g(U)$, every component $F_i$ has the sign of $g(U)$. Therefore
\begin{equation}\label{eq:F}
\sum_iF_i=g(U)\rho=U',\qquad\abs{F}_1=\abs{g(U)}\rho=\abs{U'} .
\end{equation}
The total size of the reaction term is the growth rate of the total population, and by \eqref{eq:TV} it is integrable over $[0,\infty)$. No information on how fast $g(U(t))$ tends to zero is needed.

\emph{Step 3: variation of constants.} Let $d(t)=u(t)-U(t)\pi$ be the deviation of the true solution from the stationary composition, scaled by the current total population; since $\sum_iu_i=U$ and $\sum_i\pi_i=1$, $d(t)$ is zero-sum for every $t$. Because $U(t)$ is a scalar and $A\pi=0$, $A\bigl(U(t)\pi\bigr)=0$, so
\[
d'=u'-U'\pi=\bigl(Au+F\bigr)-U'\pi=Ad+\bigl(F-U'\pi\bigr),
\]
a linear, non-homogeneous equation for $d$ with forcing $\phi:=F-U'\pi$. By the variation-of-constants formula,
\begin{equation}\label{eq:voc}
d(t)=e^{tA}d(0)+\int_0^te^{(t-s)A}\phi(s)\,ds .
\end{equation}
Two facts about $\phi$ let us invoke Lemma~\ref{lem:mix}(iii), which applies only to zero-sum vectors:
\begin{itemize}[itemsep=1pt]
\item $\phi(s)$ is zero-sum: by \eqref{eq:F}, $\sum_iF_i(s)=U'(s)$, and $\sum_i\pi_i=1$, so $\sum_i\phi_i(s)=0$;
\item its size is controlled by $\abs{U'}$: by the triangle inequality, $\abs\pi_1=1$ and \eqref{eq:F},
\[
\abs{\phi(s)}_1\le\abs{F(s)}_1+\abs{U'(s)}\,\abs\pi_1=2\abs{U'(s)}.
\]
\end{itemize}
Taking $\ell^1$ norms in \eqref{eq:voc} and applying Lemma~\ref{lem:mix}(iii) to $e^{tA}d(0)$ and to each $e^{(t-s)A}\phi(s)$, we obtain
\begin{equation}\label{eq:star}
\abs{u(t)-U(t)\pi}_1\ \le\ Ce^{-\lambda t}\abs{u(0)-U(0)\pi}_1+2C\int_0^te^{-\lambda(t-s)}\abs{U'(s)}\,ds .
\end{equation}

\emph{Step 4: the deviation vanishes.} The first term tends to zero. For the second, we split the integral at the midpoint:
\[
I_1(t):=\int_0^{t/2}e^{-\lambda(t-s)}\abs{U'(s)}\,ds,\qquad I_2(t):=\int_{t/2}^te^{-\lambda(t-s)}\abs{U'(s)}\,ds .
\]
For $I_1$: when $s\in[0,t/2]$ we have $t-s\ge t/2$, so $e^{-\lambda(t-s)}\le e^{-\lambda t/2}$, and by \eqref{eq:TV}
\[
I_1(t)\ \le\ e^{-\lambda t/2}\int_0^\infty\abs{U'(s)}\,ds\ =\ e^{-\lambda t/2}\,\abs{U_\infty-U(0)}\ \xrightarrow[t\to\infty]{}\ 0 .
\]
For $I_2$: $e^{-\lambda(t-s)}\le1$, so
\[
I_2(t)\ \le\ \int_{t/2}^\infty\abs{U'(s)}\,ds\ \xrightarrow[t\to\infty]{}\ 0,
\]
because this is the tail of the convergent integral \eqref{eq:TV}. Hence $\abs{d(t)}_1\to0$. At this point we do not yet know that $U_\infty=U^*$.

\emph{Step 5: the total population reaches $U^*$.} If $U(0)=U^*$, then $U\equiv U^*$. Suppose $U(0)<U^*$ (the case $U(0)>U^*$ is symmetric) and, \emph{for contradiction}, $U_\infty<U^*$.

{(A) If $U_\infty<U^*$, then $\int_0^\infty\rho\,dt<\infty$.} On $[U(0),U_\infty]\subset(0,U^*)$, $g$ is continuous and strictly positive by (H2), hence bounded below by some $g_{\min}>0$. Since $U(t)$ stays in this interval, \eqref{eq:U} gives $\rho\le U'/g_{\min}$, and
\[
\int_0^\infty\rho(t)\,dt\ \le\ \frac{U_\infty-U(0)}{g_{\min}}<\infty .
\]

{(B) Regardless of where $U$ converges, $\int_0^\infty\rho\,dt=\infty$.} This uses (H3)--(H4) and Step~4. Indeed, since $\rho$ is a sum of nonnegative terms, keeping only the persistently proliferating compartment $s$ of (H4), for $t\ge t_r$,
\[
\rho(t)\ \ge\ r_s(t)\,u_s(t)\ \ge\ r_{\min}\,u_s(t) . \tag{$*$}
\]
Since $u=d+U\pi$, $u_s(t)=d_s(t)+U(t)\pi_s\ge m\,\pi_s-\abs{d(t)}_1$. By Step~4, $\abs{d(t)}_1\to0$, so there is $T\ge t_r$ with $\abs{d(t)}_1\le\tfrac12m\pi_s$ for all $t\ge T$ (recall $\pi_s>0$ by Lemma~\ref{lem:mix}(iv)). Then $u_s(t)\ge\tfrac12m\pi_s$ and, by $(*)$,
\[
\rho(t)\ \ge\ r_{\min}\tfrac12m\pi_s=:c_0>0\qquad\text{for all }t\ge T,
\]
so $\int_0^\infty\rho\,dt\ge\int_T^\infty c_0\,dt=\infty$.

(A) and (B) cannot both hold, so $U_\infty=U^*$, which proves (a). In words, connected plasticity keeps the share of the proliferating compartment $s$ close to $\pi_s>0$, which keeps total proliferation from stopping before $U$ reaches $U^*$.

\emph{Step 6: conclusion.} By Steps 4 and 5, $\abs{u(t)-U^*\pi}_1\le\abs{d(t)}_1+\abs{U(t)-U^*}\to0$, which proves (b). If $u^*$ solves $u'=Au$ with total $U^*$, then $u^*(t)-U^*\pi=e^{tA}\bigl(u^*(0)-U^*\pi\bigr)$, with zero-sum initial datum, which tends to zero by Lemma~\ref{lem:mix}(iii). Hence $\abs{u(t)-u^*(t)}_1\to0$, which proves (c).
\end{proof}

\begin{remark}[Phenotypic frequencies]\label{rem:freq}
Dividing $d(t)$ by $U(t)\ge m$ gives the same estimate for the frequencies $p=u/U$: $\abs{p(t)-\pi}_1\le m^{-1}\abs{d(t)}_1$. Equivalently, $p$ obeys $p'=Ap+g(U)\bigl(R-\bar r\bigr)p$ with $\bar r=\sum_ir_ip_i$: the switching dynamics perturbed by a replicator-type term with zero sum, whose size is at most $\abs{g(U)}\sum_i(r_i+\bar r)p_i=2\abs{g(U)}\bar r=2\abs{(\ln U)'}$. Since $U$ is monotone, the total selection that the frequencies undergo is limited, at most $\int_0^\infty2\abs{(\ln U)'}\,dt=2\abs{\ln(U^*/U(0))}$.
\end{remark}

\begin{remark}[Sustained oscillations under weighted competition]\label{rem:oscillations} The following counterexample shows that assumption (H1) is necessary for the asymptotic equivalence in Theorem~\ref{thm:disc}. Take three states with cyclic switching $1\to2\to3\to1$ at rate $0.2$ and weak reverse switching at rate $0.001$, proliferation rates $r=(30,\,0.01,\,0.01)$, pressures $w=(0.001,\,0.001,\,1)$ and $g(W)=1-W$. All hypotheses (H2)--(H4) hold, and $\pi=(\tfrac13,\tfrac13,\tfrac13)$. The only positive equilibrium is $u_{\rm eq}=\tfrac{500}{501}(1,1,1)$, at which $W=1$, the zero of $g$, while $U=1500/501\approx3$. The Jacobian matrix at $u_{\rm eq}$ is $A-(Ru_{\rm eq})\,w^{\top}$, whose characteristic polynomial $\mu^3+a_1\mu^2+a_2\mu+a_3$ has, in exact arithmetic,
\[
a_1=\tfrac{80527}{125250}\approx0.643,\qquad a_2=\tfrac{29155041}{167000000}\approx0.175,\qquad a_3=\tfrac{60341701}{50000000}\approx1.207 .
\]
All coefficients are positive but $a_1a_2-a_3\approx-1.09<0$, so by the Routh--Hurwitz criterion \cite{gantmacher1959} two eigenvalues have positive real part (numerically $0.310\pm0.928\,i$; the third is $-1.262$). The equilibrium is therefore unstable. Solutions cannot approach $0$ either, since for small populations $W<1$ and the population grows. Consequently, solutions starting near $u_{\rm eq}$ (off a single curve of initial conditions) never settle. Numerically, they approach a stable periodic orbit of period $\approx15.5$ along which the total population oscillates between $\approx1.5$ and $\approx43$ (Figure~\ref{fig:cycle}).

The mechanism is a \emph{delayed negative feedback}. State~1 proliferates fast but exerts almost no competitive pressure, while state~3, which proliferates slowly, exerts almost all of it. New cells produced in state~1 raise the pressure only after they have crossed the chain $1\to2\to3$, which takes about $10$ time units, comparable to the period. Delayed negative feedback is a classical source of oscillations in physiological control systems, in particular in regulated hematopoiesis \cite{mackeyglass1977,mackey1978}. Two heterogeneities are needed at the same time: if all compartments proliferate at the same rate, the composition obeys $p'=Ap$ exactly and converges to $\pi$; if all pressures are equal, we are back in Theorem~\ref{thm:disc}. In the language of the proof, what is lost is conservation: switching conserves the number of cells $U$, but it does not conserve the regulating quantity $W$, because a cell that switches phenotype also changes its competitive weight. In simulations with fast, well-connected switching we always observed convergence to $u_{\rm eq}$; whether this holds in general is open.
\end{remark}

\begin{figure}[htbp]
\centering
\includegraphics[width=0.98\textwidth]{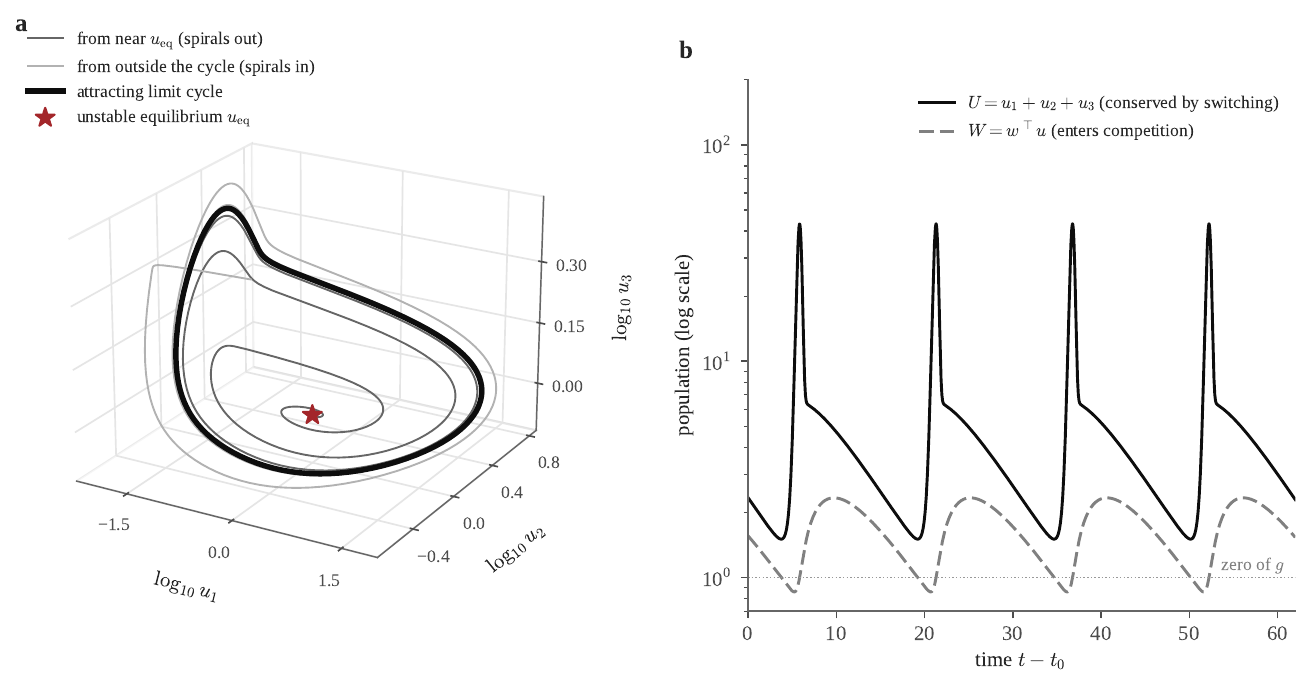}
\caption{\footnotesize\linespread{1}\selectfont\textbf{Weighted competition can destroy convergence (counterexample of Appendix~A).} \textbf{(a)} Phase space $(\log_{10}u_1,\log_{10}u_2,\log_{10}u_3)$: a trajectory starting close to the unstable equilibrium $u_{\rm eq}$ (star) spirals outwards, a trajectory starting outside spirals inwards, and both approach the same attracting limit cycle (thick line), of period $\approx15.5$. \textbf{(b)} Along the cycle, the weighted population $W=\sum_jw_ju_j$, which regulates growth, oscillates indefinitely across $1$, the zero of $g$, while the total population $U$, which is conserved by switching, oscillates between $\approx1.5$ and $\approx43$ (its equilibrium value is $1500/501\approx3$); four periods are shown.}
\label{fig:cycle}
\end{figure}

\subsection*{B. Method of lines and discretization of the advection--diffusion PDE}

We show that the compartmental model with local transitions \eqref{eq:modelondloc} is precisely a finite-difference discretization of the continuous reaction--diffusion--advection equation
\begin{equation}\label{eq:ap_pde}
    \partial_t u = f(u,x,t) - \partial_x(v u) + \partial_x(D\,\partial_x u),
\end{equation}
and we derive the identification formulas
\[
    D(x_{i+1/2}) \approx \frac{\Delta x^2}{2}\bigl(k_{i,i+1}+k_{i+1,i}\bigr),\qquad
    v(x_{i+1/2}) \approx \Delta x\,\bigl(k_{i,i+1}-k_{i+1,i}\bigr),
\]
and comment on the consistency of the scheme.

\medskip
\noindent\textbf{Grid and notation.}
Consider a bounded one-dimensional domain \(\Omega = [0,L]\). Partition it into \(n\) cells of equal width \(\Delta x = L/n\), with centers
\[
    x_i = \bigl(i-\tfrac12\bigr)\Delta x,\quad i=1,\dots,n,
\]
and interfaces
\[
    x_{i+1/2} = i\,\Delta x,\quad i=0,\dots,n,
\]
so that the outer interfaces \(x_{1/2}=0\) and \(x_{n+1/2}=L\) are the boundaries of the domain. We denote \(u_i(t) \approx u(x_i,t)\), \(v_i \approx v(x_i)\), \(D_{i+1/2} \approx D(x_{i+1/2})\), and \(f_i(u):=f(u,x_i,t)\).

\medskip
\noindent\textbf{A conservative (finite-volume) formulation.}
Integrate the PDE \eqref{eq:ap_pde} exactly over the control volume \([x_{i-1/2},x_{i+1/2}]\) surrounding node \(x_i\). Writing the total flux as \(J(x,t):=v(x,t)u(x,t)-D(x,t)\partial_xu(x,t)\), so that \eqref{eq:ap_pde} reads \(\partial_t u = f-\partial_xJ\), the fundamental theorem of calculus gives,
\begin{equation}\label{eq:ap_fv_exact}
    \frac{d}{dt}\int_{x_{i-1/2}}^{x_{i+1/2}} u(x,t)\,dx = \int_{x_{i-1/2}}^{x_{i+1/2}} f(u,x,t)\,dx \;-\;\bigl[J(x_{i+1/2},t)-J(x_{i-1/2},t)\bigr].
\end{equation}
This is the defining property of a \emph{conservative} discretization: whatever leaves cell \(i\) through an interface enters the neighboring cell, so that, summing over all cells, only the fluxes through the domain boundaries survive, as in the PDE itself.

Approximating the cell average by the nodal value, \(\frac{1}{\Delta x}\int_{x_{i-1/2}}^{x_{i+1/2}}u\,dx\approx u_i(t)\), \eqref{eq:ap_fv_exact} becomes, per unit length,
\begin{equation}\label{eq:ap_fv_semi}
    \frac{du_i}{dt} = f_i(u) - \frac{J(x_{i+1/2},t)-J(x_{i-1/2},t)}{\Delta x}.
\end{equation}
It remains to approximate the advective and diffusive parts of the flux at the interfaces.

\medskip
\noindent\textbf{Advective flux at the interface.}
We approximate \(J_{\text{adv}}(x,t)=v(x,t)u(x,t)\) at \(x_{i+1/2}\) by evaluating \(v\) exactly there and averaging \(u\) between its two neighboring nodes,
\begin{equation}\label{eq:ap_Jadv}
    J_{\text{adv}}(x_{i+1/2},t) \approx v_{i+1/2}\,\frac{u_i+u_{i+1}}{2}.
\end{equation}
Substituting \eqref{eq:ap_Jadv} (and its analogue at \(x_{i-1/2}\)) into the advective part of \eqref{eq:ap_fv_semi} gives
\begin{equation}\label{eq:ap_adv_disc}
    -\frac{J_{\text{adv}}(x_{i+1/2})-J_{\text{adv}}(x_{i-1/2})}{\Delta x}
    \approx \frac{v_{i-1/2}}{2\Delta x}u_{i-1}
    + \frac{v_{i-1/2} - v_{i+1/2}}{2\Delta x}u_i
    - \frac{v_{i+1/2}}{2\Delta x}u_{i+1}.
\end{equation}

\medskip
\noindent\textbf{Diffusive flux at the interface.}
We approximate \(J_{\text{diff}}(x,t)=-D(x,t)\partial_xu(x,t)\) at \(x_{i+1/2}\) by evaluating \(D\) exactly there and the derivative by a centered difference,
\begin{equation}\label{eq:ap_Jdiff}
    J_{\text{diff}}(x_{i+1/2},t)\approx -D_{i+1/2}\,\frac{u_{i+1}-u_i}{\Delta x}.
\end{equation}
Substituting \eqref{eq:ap_Jdiff} (and its analogue at \(x_{i-1/2}\)) into the diffusive part of \eqref{eq:ap_fv_semi} gives
\begin{equation}\label{eq:ap_diff_disc}
    -\frac{J_{\text{diff}}(x_{i+1/2})-J_{\text{diff}}(x_{i-1/2})}{\Delta x}
    = \frac{D_{i-1/2}}{\Delta x^2}u_{i-1}
       - \frac{D_{i-1/2}+D_{i+1/2}}{\Delta x^2}u_i
       + \frac{D_{i+1/2}}{\Delta x^2}u_{i+1},
\end{equation}
the standard three-point stencil for the Laplacian with variable coefficients.

\medskip
\noindent\textbf{Assembling the semi-discrete system.}
Adding \eqref{eq:ap_adv_disc} and \eqref{eq:ap_diff_disc} to \eqref{eq:ap_fv_semi}, the PDE \eqref{eq:ap_pde} at interior grid points becomes the ODE system
\begin{align}\label{eq:ap_ode_interior}
    \frac{du_i}{dt} = &
     f_i(u)
    + \left( \frac{D_{i-1/2}}{\Delta x^2} + \frac{v_{i-1/2}}{2\Delta x} \right) u_{i-1} + \left( \frac{D_{i+1/2}}{\Delta x^2} - \frac{v_{i+1/2}}{2\Delta x} \right) u_{i+1} \notag\\
    &
    - \left( \frac{D_{i-1/2}+D_{i+1/2}}{\Delta x^2} + \frac{v_{i+1/2} - v_{i-1/2}}{2\Delta x} \right) u_i.
\end{align}

\medskip
\noindent\textbf{Connection with the compartmental model.}
The local-transition model \eqref{eq:modelondloc} for an interior compartment \(i\) is
\begin{equation}\label{eq:ap_comp_interior}
    u_i' = r_i(t) u_i g_i(u,t) + k_{i-1,i} u_{i-1} + k_{i+1,i} u_{i+1} - (k_{i,i-1}+k_{i,i+1}) u_i.
\end{equation}
Assuming \(f_i(u)=r_i(t)u_ig_i(u,t)\), we identify \eqref{eq:ap_ode_interior} and \eqref{eq:ap_comp_interior} term by term. Matching the \(u_{i-1}\) coefficient, which involves only the left interface \(x_{i-1/2}\), gives
\begin{equation}\label{eq:ap_match_left}
    k_{i-1,i} = \frac{D_{i-1/2}}{\Delta x^2} + \frac{v_{i-1/2}}{2\Delta x};
\end{equation}
matching the \(u_{i+1}\) coefficient, which involves only the right interface \(x_{i+1/2}\), gives
\begin{equation}\label{eq:ap_match_right}
    k_{i+1,i} = \frac{D_{i+1/2}}{\Delta x^2} - \frac{v_{i+1/2}}{2\Delta x}.
\end{equation}
Since \(x_{i+1/2}\) is also the \emph{left} interface of compartment \(i+1\), relabeling \(i\to i+1\) in \eqref{eq:ap_match_left} gives the complementary (outward) rate at that same interface,
\begin{equation}\label{eq:ap_match_right_shifted}
    k_{i,i+1} = \frac{D_{i+1/2}}{\Delta x^2} + \frac{v_{i+1/2}}{2\Delta x};
\end{equation}
similarly, relabeling \(i\to i-1\) in \eqref{eq:ap_match_right} gives \(k_{i,i-1}=D_{i-1/2}/\Delta x^2 - v_{i-1/2}/(2\Delta x)\), and \(-(k_{i,i-1}+k_{i,i+1})\) then reproduces the coefficient of \(u_i\) in \eqref{eq:ap_ode_interior}: the interface values \((D_{i\pm1/2},v_{i\pm1/2})\) determine all three coefficients.

The rates \(k_{i,i+1}\) (forward) and \(k_{i+1,i}\) (backward) in \eqref{eq:ap_match_right_shifted} and \eqref{eq:ap_match_right} both live at the same interface \(x_{i+1/2}\), so they invert for \(D_{i+1/2}\) and \(v_{i+1/2}\). Adding them eliminates \(v\):
\begin{equation}\label{eq:ap_D_solved}
    D_{i+1/2} = \frac{\Delta x^2}{2} \bigl( k_{i,i+1} + k_{i+1,i} \bigr);
\end{equation}
subtracting them eliminates \(D\):
\begin{equation}\label{eq:ap_v_solved}
    v_{i+1/2} = \Delta x \bigl( k_{i,i+1} - k_{i+1,i} \bigr).
\end{equation}
These are exactly the relations \eqref{eq:escolhaDv} quoted in the main text.

\medskip
\noindent\textbf{The symmetric and skew-symmetric parts, \(\mathcal{D}\) and \(\mathcal{V}\).}
The tridiagonal matrix \(A\) has entries \(A_{i,i-1}=k_{i-1,i}\), \(A_{i,i+1}=k_{i+1,i}\), \(A_{ii}=-(k_{i,i-1}+k_{i,i+1})\). Write \(A=A_{\rm sym}+A_{\rm skew}\) with \(A_{\rm sym}=(A+A^\top)/2\), \(A_{\rm skew}=(A-A^\top)/2\), and define \(\mathcal D:=\Delta x^2A_{\rm sym}\), \(\mathcal V:=2\Delta xA_{\rm skew}\), so that \(A=\Delta x^{-2}\mathcal D+(2\Delta x)^{-1}\mathcal V\), as in \eqref{eq:matrixDecomp}. Using \eqref{eq:ap_match_left}--\eqref{eq:ap_match_right_shifted} and its \(i\to i-1\) shift, a direct computation gives, for the off-diagonal entries,
\[
(\mathcal D)_{i,i\mp1} = \frac{\Delta x^2}{2}(k_{i\mp1,i}+k_{i,i\mp1}) = D_{i\mp1/2}, \qquad
(\mathcal V)_{i,i\mp1} = \Delta x\,(k_{i\mp1,i}-k_{i,i\mp1}) = \pm\, v_{i\mp1/2}:
\]
the off-diagonal entries of \(\mathcal D\) and \(\mathcal V\) are, exactly, the interface values of \(D\) and \(v\), confirming that the symmetric part carries diffusion and the skew-symmetric part carries advection. The diagonal, however, does not split so cleanly. Since \(\mathcal V\) is antisymmetric by construction, \((\mathcal V)_{ii}\equiv0\) identically, for \emph{any} \(D(x)\) and \(v(x)\); correspondingly, the entire diagonal of \(A\), and hence of \(\mathcal D\), must fall on the symmetric side:
\[
(\mathcal D)_{ii} = -(D_{i-1/2}+D_{i+1/2}) - \frac{\Delta x}{2}\bigl(v_{i+1/2}-v_{i-1/2}\bigr).
\]
The first term is the familiar three-point Laplacian stencil, \((\mathcal D)_{ii}\to-2D(x_i)\) as \(\Delta x\to0\). The second term has nothing to do with diffusion, and, although it is of order \(\Delta x^2\) relative to the first, it does not vanish in the limit: after division by \(\Delta x^2\) in \eqref{eq:matrixDecomp} it contributes \(-(v_{i+1/2}-v_{i-1/2})u_i/(2\Delta x)\to-\tfrac12v'(x_i)u_i\). It is half of the compressibility term of the advective flux, \(-\partial_x(vu)=-v\,\partial_xu-v'u\); the other half is carried by the skew-symmetric part, which approximates \(-v\,\partial_xu-\tfrac12v'u\). The symmetric/skew-symmetric split therefore separates diffusion and advection exactly when \(v\) is uniform; in general, the symmetric part carries diffusion plus the local term \(-\tfrac12v'u\).

\medskip
\noindent\textbf{Two readings of the identification.}
Relations \eqref{eq:ap_D_solved}--\eqref{eq:ap_v_solved} can be read in two directions. Given smooth functions \(D(x)>0\) and \(v(x)\), they \emph{define}, for each grid spacing \(\Delta x\), a compartmental model that is a consistent discretization of \eqref{eq:ap_pde}; this is the direction used in the simulation of Figure~\ref{fig:pdesim}. Conversely, a compartmental model with many nearest-neighbour states approximates a continuous dynamics only if its rates scale appropriately: \(k_{i,i+1}\) and \(k_{i+1,i}\) must diverge as \(O(\Delta x^{-2})\), with a difference of order \(O(\Delta x^{-1})\), the familiar scaling of a biased random walk converging to a drift--diffusion process; rates held fixed as \(\Delta x\to0\) collapse both \(D\) and \(v\) to zero.

\medskip
\noindent\textbf{Boundary conditions.}
The no-flux boundary conditions at \(x=0\) and \(x=L\) are discretized in the same conservative spirit, by setting the flux through the two boundary interfaces to zero rather than approximating it from interior values: \(J(x_{1/2},t)=J(x_{n+1/2},t)=0\). For the first cell, \eqref{eq:ap_fv_semi} then reads
\[
    \frac{du_1}{dt} = f_1(u) - \frac{J(x_{3/2},t)}{\Delta x}
    \approx f_1(u) + \Bigl(\frac{D_{3/2}}{\Delta x^2}-\frac{v_{3/2}}{2\Delta x}\Bigr)u_2 - \Bigl(\frac{D_{3/2}}{\Delta x^2}+\frac{v_{3/2}}{2\Delta x}\Bigr)u_1,
\]
which is the compartmental equation for \(u_1\), with the rates \eqref{eq:ap_match_right} and \eqref{eq:ap_match_right_shifted} at the interface \(x_{3/2}\) and no transitions through \(x_{1/2}\), i.e., \(k_{0,1}=k_{1,0}=0\). The same holds for the last cell, with \(k_{n,n+1}=k_{n+1,n}=0\). Hence the compartmental model with no transitions out of \(u_1\) and \(u_n\) is a faithful discretization of the PDE with zero-flux boundary conditions.

\medskip
\noindent\textbf{Consistency and positivity.}
Both flux approximations, \eqref{eq:ap_Jadv} and \eqref{eq:ap_Jdiff}, are centered at the interfaces, so, for smooth \(u\), \(D\) and \(v\), Taylor expansion shows that \eqref{eq:ap_ode_interior} approximates the right-hand side of \eqref{eq:ap_pde} at \(x_i\) with a local truncation error \(O(\Delta x^2)\); together with the stability of centered finite-volume discretizations of parabolic equations \cite{LeVeque}, the scheme converges to \eqref{eq:ap_pde} as \(\Delta x\to0\) while the solution stays smooth. For the identification to define a legitimate compartmental (Markov chain) model, the rates \eqref{eq:ap_match_right}--\eqref{eq:ap_match_right_shifted} must be nonnegative, which holds if \(\Delta x\le 2D_{i+1/2}/|v_{i+1/2}|\), a mesh (cell) P\'eclet-number condition; otherwise an upwind treatment of the advective flux restores positivity, at the cost of first-order accuracy.

\subsection*{Acknowledgements}
\addcontentsline{toc}{section}{Acknowledgements}
The author thanks Prof.~Ingmar Glauche for suggestions on an early draft of this manuscript. Funding: This work was supported by CAPES and CNPq (Brazil), the Alexander von Humboldt Foundation (Germany), and partially by FAPEMIG (Brazil).

\subsection*{Declaration of competing interest}
\addcontentsline{toc}{section}{Declaration of competing interest}
The author declares no competing interests.

\subsection*{Data availability}
\addcontentsline{toc}{section}{Data availability}
This is a theoretical study; no new data were generated or analyzed. The Python code that generates the figures, including the control points of the multi-well landscape, is available at \url{https://github.com/arturfassoni/sgd-epigenetic-landscape-review}.

\subsection*{Declaration of generative AI and AI-assisted technologies in the manuscript preparation process}
\addcontentsline{toc}{section}{Declaration of generative AI and AI-assisted technologies in the manuscript preparation process}
During the preparation of this work the author used Claude (Anthropic) in order to search and verify parts of the scientific literature, and to assist with translation and language editing. After using this tool, the author reviewed and edited the content as needed and takes full responsibility for the content of the published article.


\end{document}